\documentclass[11pt, a4paper]{amsart}
\usepackage[dvips]{epsfig}
\usepackage{amsmath}
\usepackage{amssymb}
\usepackage{amsfonts}
\usepackage{amsthm}
\usepackage{amsbsy}
\usepackage{amsgen}
\usepackage{amscd}
\usepackage{amsopn}
\usepackage{amstext}
\usepackage{amsxtra}
\usepackage{enumerate}
\usepackage{dsfont}
\usepackage{hyperref}
\usepackage{lineno}
\usepackage{times, graphicx}
\usepackage{eso-pic}
\usepackage{wrapfig}
\usepackage{blindtext}

\usepackage{datetime}

\newtheorem{theorem}{Theorem}[section]
\newtheorem{proposition}[theorem]{Proposition}
\newtheorem{lemma}[theorem]{Lemma}
\newtheorem{corollary}[theorem]{Corollary}
\theoremstyle{definition}
\newtheorem{definition}[theorem]{Definition}
\newtheorem{notation}[theorem]{Notation}
\newtheorem{remark}[theorem]{Remark}
\newtheorem{conjecture}[theorem]{Conjecture}

\newtheorem{question}[theorem]{Question}

\newcommand{\vol}{\operatorname{Vol}}
\newcommand{\prob}{\mathcal{P}}

\newcommand {\T} {\mathbb{T}}

\newcommand {\E} {\mathbb{E}}
\newcommand {\M} {\mathcal{M}}
\newcommand {\Ac} {\mathcal{A}}
\newcommand {\Dc} {\mathcal{D}}

\newcommand {\Sc} {\mathcal{S}}
\newcommand {\R} {\mathbb{R}}
\newcommand {\Z} {\mathbb{Z}}

\newcommand {\Cc} {\mathcal{C}}
\newcommand {\Pc} {\mathcal{P}}
\newcommand {\Pzero} {\mathcal{P}_0}
\newcommand {\Ptorus} {\mathcal{P}_{\text{symm}}}
\newcommand {\Peps} {\mathcal{P}_{\epsilon}}

\newcommand {\Lc} {\mathcal{L}}
\newcommand {\Zc} {\mathcal{Z}}
\newcommand {\Tb} {\mathbb{T}}
\newcommand {\Fc} {\mathcal{F}}

\newcommand {\area} {\operatorname{Area}}
\newcommand {\var} {\operatorname{Var}}
\newcommand {\cov} {\operatorname{Cov}}

\newcommand {\Nc} {\mathcal{N}}
\newcommand {\Bc} {\mathcal{B}}

\newcommand{\fixmenot}[1]{}
\newcommand{\fixmeparhidden}[1]{}
\newcommand{\hide}[1]{\footnote{Hidden.}}

\numberwithin{equation}{section}

\begin{document}

%\linenumbers

%\title[Non-universality of the Nazarov-Sodin
%constant]{Non-universality of the Nazarov-Sodin constant for
%monochromatic waves and toral eigenfunctions}
\title[Variation the Nazarov-Sodin constant]{Variation of the Nazarov-Sodin constant for random plane waves and arithmetic random waves}

\author{P\"ar Kurlberg}
\author{Igor Wigman}

% \date{August 7, 2017}
%\date{\pknew{March 1, 2018}}

%\today, { } \currenttime{ }

\begin{abstract}
  This is a manuscript containing the full proofs of results announced
  in ~\cite{KuWi announcment}, together with some recent updates.  We
  prove that the {\em Nazarov-Sodin constant}, which up to a natural
  scaling gives the leading order growth for the expected number of
  nodal components of a random Gaussian field, genuinely depends on
  the field. We then infer the same for ``arithmetic random waves'',
  i.e. random toral Laplace eigenfunctions.
\end{abstract}

\maketitle

\section{Introduction}

%\subsection{The Nazarov-Sodin constant}

For $m\ge 2$, let $f:\R^{m}\rightarrow\R$ be a {\em stationary}
centred Gaussian random field, and $r_{f}:\R^{m}\rightarrow \R$ its
covariance function
$$
%r_{f}(x)= \E[f(y)f(y+x)].
r_{f}(x)= \E[f(0 )\cdot f(x)] = \E[f(y)\cdot f(x+y)].
$$
Given such an $f$, let $\rho$ denote its spectral measure, i.e. the
Fourier transform of $r_{f}$ (assumed to be a probability measure);
note that prescribing $\rho$ uniquely defines $f$ by Kolmogorov's
Theorem (cf. \cite[Chapter~3.3]{cramer-leadbetter}.) In
what follows we often allow for $\rho$ to vary; it will be convenient
to let $f_{\rho}$ denote a random field with spectral measure $\rho$.
We further assume that a.s. $f_{\rho}$ is sufficiently smooth, and
that the distribution of $f$ and its derivatives is non-degenerate in
an appropriate sense (a condition on the support of $\rho$).

\subsection{Nodal components and the Nazarov-Sodin constant}

Let $\Nc(f_{\rho};R)$ be the number of connected components of
$f_{\rho}^{-1}(0)$ in $B_{0}(R)$ (the radius-$R$ ball centred at $0$),
usually referred to as the {\em nodal components} of $f_{\rho}$;
$\Nc(f_{\rho};R)$ is a random variable. Assuming further that
$f_{\rho}$ is ergodic (equivalently, $\rho$ has no atoms), with
non-degenerate gradient
distribution (equivalent to $\rho$ not being supported on a hyperplane
passing through the origin), Nazarov and Sodin \cite[Theorem $1$]{So}
evaluated the expected number of nodal components of $f_{\rho}$ to be
asymptotic to
\begin{equation}
\label{eq:EN(f;R)=cNSR^m+o(Rm)}
\E[\Nc(f_{\rho};R)] = c_{NS}(\rho)\cdot \vol(B(1))\cdot R^{m}+o(R^{m}),
\end{equation}
where $c_{NS}(\rho) \ge 0$ is a constant, subsequently referred to as
the ``Nazarov-Sodin
constant'' of $f_{\rho}$, and $\vol(B(1))$ is the volume of the
radius-$1$ $m$-ball $B(1)\subseteq\R^{m}$.
They also established convergence in mean, i.e., that
\begin{equation}
\label{eq:|EN(f)-xNS|->0}
\E\left[\left| \frac{\Nc(f_{\rho};R)}{\vol(B(1))\cdot R^{m}} -
    c_{NS}(\rho)  \right|\right] \rightarrow 0;
\end{equation}
it is a consequence of the assumed ergodicity of
the underlying random field $f_{\rho}$. In this manuscript we will
consider $c_{NS}$ as a function of the spectral density $\rho$,
without assuming that $f_{\rho}$ is ergodic. To our best knowledge,
the value of $c_{NS}(\rho)$, even for a single $\rho$, was not rigorously
known heretofore.

For $m=2$, $\rho=\rho_{\mathcal{S}^{1}}$ the uniform measure on the
unit circle $\mathcal{S}^{1}\subseteq \R^{2}$
(i.e. $d\rho=\frac{d\theta}{2\pi}$ on $\mathcal{S}^{1}$, and vanishing
outside the circle) the corresponding random field $f_{\text{RWM}}$ is
known as {\em random monochromatic wave}; according to Berry's {\em Random Wave Model}
~\cite{Berry 1977}, $f_{\text{RWM}}$ serves as
a universal model for Laplace eigenfunctions on generic surfaces in
the high energy limit. The corresponding {\em universal}
Nazarov-Sodin constant
\begin{equation}
\label{eq:cRWM>0}
c_{\text{RWM}}=c_{NS}\left(\frac{d\theta}{2\pi}\right)>0
\end{equation}
was proven to be strictly
positive ~\cite{NS 2009}. Already in \cite{BS}, Bogomolny and Schmit
employed the percolation theory to predict its value, but recent
numerics by Nastacescu \cite{Na}, Konrad \cite{Ko} and
Beliaev-Kereta~\cite{BK}, consistently indicate a $4.5-6\%$
deviation from
these predictions.

\vspace{3mm}

More generally, let $(\mathcal{M}^{m},g)$ be a smooth compact
Riemannian manifold of volume $\vol(\M)$. Here the restriction of a
fixed random field $f:\mathcal{M}\rightarrow \R$ to growing domains,
as was considered on the Euclidean space, makes no sense. Instead we
consider a sequence of smooth non-degenerate random fields
$\{f_{L} \}_{L\in \mathcal{L}}$ (for $\mathcal{L}\subseteq\R$ some
discrete subset), and the total number $\Nc(f_{L})$ of nodal
components of $f_{L}$ on $\M$ (the case $\M = \T^{2}$ will be treated
in \S~\ref{sec:statement-toral}; $\mathcal{L}$ will then be a
subset of the Laplace spectrum for $\T^{2}$). Here we may define a
scaled covariance function of $f_{L}$ around a fixed point
$x\in\mathcal{M}$ on its tangent space
$T_{x}(\mathcal{M})\cong \R^{m}$ via the exponential map at $x$, and
assume that for a.e. $x\in\mathcal{M}$ the scaled covariance and a few of its derivatives
converge, locally uniformly, to the covariance function of a limiting
stationary Gaussian field around $x$ and its respective derivatives;
let $\rho_{x}$ be the corresponding spectral density. For the setup as above, Nazarov-Sodin
proved ~\cite[Theorem $4$]{So} that as $L\to \infty$,
\begin{equation*}
\E [ \Nc(f_{L}) ] = \overline{c_{NS}}\cdot \vol(\M)\cdot L^{m}+o(L^{m}),
\end{equation*}
for some $\overline{c_{NS}}\ge 0$ depending on the limiting fields
only, or, more precisely,
$$\overline{c_{NS}} = \int\limits_{\mathcal{M}}c_{NS}(\rho_{x})dx$$ is the superposition of their Nazarov-Sodin constants.
This result applies in
particular to random band-limited functions on a generic Riemannian
manifold, considered in ~\cite{SW}, with the constant
$\overline{c_{NS}}>0$ strictly positive.

\subsection{Statement of results for random waves on $\R^{2}$}
\label{sec:res scale inv}

Let $\Pc$ be the collection of probability measures on
$\R^{2}$ supported on the radius-$1$ standard ball
$B(1)\subseteq \R^{2}$, and invariant under rotation by $\pi$. We
note that any spectral measure can without loss of generality be assumed
to be $\pi$-rotation invariant, hence the collection of spectral measures
supported on $B(R)$ can, after rescaling, be assumed to lie in $\mathcal{P}$.

Our first goal (Proposition \ref{prop:NS constant uniform} below)
is to extend the definition of the Nazarov-Sodin
constant for all $\rho \in \mathcal{P}$, in particular, we allow
spectral measures possessing atoms. We show that one may define
$c_{NS}$ on $\Pc$ such that the defining property
\eqref{eq:EN(f;R)=cNSR^m+o(Rm)} of $c_{NS}$ is satisfied, though its
stronger form \eqref{eq:|EN(f)-xNS|->0} might not necessarily hold. Further, the
limit on the l.h.s. of \eqref{eq:|EN(f)-xNS|->0} always exists, even
if it is not vanishing (Proposition \ref{prop:NS constant uniform part two} below, cf. \S~\ref{sec:example-spectr-meas}).
Rather than counting the nodal components
lying in discs of increasing radius, we will count components
lying in squares with increasing side lengths; by abuse of notation
from now on $\Nc(f_{\rho};R)$ will denote the number of nodal
components of $f_{\rho}$ lying in the square
\begin{equation*}
\Dc_{R} := [-R,R]^{2} \subseteq \R^2.
\end{equation*}
Though the results are equivalent for both settings (every result we
are going to formulate on domains lying in squares could equivalently
be formulated for discs), unlike discs, the squares possess the
extra-convenience of tiling into smaller squares. This obstacle could
be easily mended for the discs using the ingenious
``Integral-Geometric Sandwich'' (which can be viewed as an
infinitesimal tiling) introduced by Nazarov-Sodin ~\cite{So}.

\begin{proposition}
\label{prop:NS constant uniform}
Let $f_{\rho}$ be a plane random field with spectral density $\rho\in\mathcal{P}$.
The limit
\begin{equation*}
c_{NS}(\rho) := \lim\limits_{R\rightarrow\infty}\frac{\E[\Nc(f_{\rho};R)]}{4R^{2}}
\end{equation*}
exists and is uniform w.r.t. $\rho\in\Pc$. More precisely, for every $\rho\in\Pc$ we have
\begin{equation}
\label{eq:N(R)=cR^2+O(R)}
\E[\Nc(f_{\rho};R)] = c_{NS}(\rho)\cdot 4 R^{2} + O(R)
\end{equation}
with constant involved in the `$O$'-notation {\em absolute}.
\end{proposition}

As for fluctuations around the mean \`{a} la \eqref{eq:|EN(f)-xNS|->0}, we have the following result:
\begin{proposition}
\label{prop:NS constant uniform part two}
%\item
The limit
\begin{equation}
\label{eq:dNS lim abs def}
d_{NS}(\rho):=\lim\limits_{R\rightarrow\infty}\E\left[\left|\frac{\Nc(f_{\rho};R)}{4R^{2}} -c_{NS}(\rho)\right|\right]
\end{equation}
(``Nazarov-Sodin discrepancy functional") exists for all $\rho\in\Pc$.

\end{proposition}

However,
the limit \eqref{eq:dNS lim abs def} is not uniform w.r.t.
$\rho \in \Pc$, so in particular, an analogue of \eqref{eq:N(R)=cR^2+O(R)} does not
hold for $d_{NS}(\cdot)$. For had \eqref{eq:dNS lim abs def} been uniform,
a proof similar to the proof of Theorem \ref{thm:cNS cont} below would yield the continuity of $d_{NS}(\cdot)$;
this cannot hold\footnote{We are greatful to Dmitry Beliaev for pointing it out to us},
since on one hand it is possible to construct
a measure $\rho\in\Pc$ with $d_{NS} > 0$ (see \S~\ref{sec:example-spectr-meas}),
and on the other hand it is possible to approximate an arbitrary measure $\rho\in\Pc$ with a smooth
one $\rho'$ (e.g. by convolving with smooth mollifiers), so that $f_{\rho'}$ is ergodic,
and $d_{NS}(\rho') = 0$.

We believe that the
uniform rate of convergence \eqref{eq:N(R)=cR^2+O(R)} is of two-fold independent
interest. First, for numerical simulations it determines the value of
sufficiently big radius $R$ to exhibit a realistic nodal portrait with
the prescribed precision. Second, it is instrumental for the proof of
Theorem \ref{thm:cNS cont} below, a principal result of this
manuscript.

\begin{theorem}
\label{thm:cNS cont}
The map $c_{NS}:\mathcal{P}\rightarrow\R_{\ge 0}$, given by
$$c_{NS}:\rho\mapsto c_{NS}(\rho)$$
is a continuous functional
w.r.t. the weak-* topology on $\mathcal{P}$.
\end{theorem}

To prove Theorem \ref{thm:cNS cont} we follow the steps of
Nazarov-Sodin ~\cite{So} closely, controlling the various error
terms encountered. One of the key aspects of our proof, different from
Nazarov-Sodin's, is the uniform version \eqref{eq:N(R)=cR^2+O(R)} of
\eqref{eq:EN(f;R)=cNSR^m+o(Rm)}.

\vspace{3mm}

Giving good lower bounds on $c_{NS}(\rho)$ appears difficult and
it is not a priori clear that $c_{NS}(\rho)$ genuinely varies with
$\rho$. However, it is straightforward to see that $c_{NS}(\rho)=0$
if $\rho$ is a delta measure supported at zero, and we can also construct
examples of monochromatic random waves with $c_{NS}(\rho)=0$ when
$\rho$ is supported on two antipodal points. (See
\S~\ref{sec:statement-toral} for some further examples of
measures $\rho$ satisfying stronger symmetry assumptions, yet with the
property that $c_{NS}(\rho)=0$.)  This, together with the convexity
and compactness of $\mathcal{P}$, easily gives the following
consequence of Theorem \ref{thm:cNS cont}.
\begin{corollary}
\label{cor:NS constant vary Eucl}
The Nazarov-Sodin constant $c_{NS}(\rho)$ for $\rho\in \mathcal{P}$
attains all values in an interval of the form $[0,c_{\max}]$ for some
$0<c_{\max}<\infty$.
\end{corollary}

Corollary \ref{cor:NS constant vary Eucl} sheds no light
%as for the mystery of
on the value of $c_{\max}$; see \S~\ref{sec:discussion} for
some
intuition and related conjectures.

\subsection{Statement of results for toral eigenfunctions (arithmetic random waves)}
\label{sec:statement-toral}

Let $S$ be the set of all integers that admit a representation as a
sum of two integer squares, and let $n\in S$. The toral Laplace
eigenfunctions
$f_{n}:\Tb^{2}=\R^{2}/\Z^{2}\rightarrow \R$ of eigenvalue $-4\pi^{2}n$ may be
expressed as
\begin{equation}
\label{eq:fm def}
f_{n}(x) = \sum\limits_{\substack{\|\lambda\|^{2}=n \\ \lambda\in\Z^{2}}}a_{\lambda}e^{2\pi i \langle x,\lambda\rangle}
\end{equation}
for some complex coefficients $\{a_{\lambda}\}_\lambda$ satisfying
$a_{-\lambda}=\overline{a_{\lambda}}$.
We endow the space of
eigenfunctions with a Gaussian probability measure by making the
coefficient $a_{\lambda}$ i.i.d. standard Gaussian (save for the
relation $a_{-\lambda}=\overline{a_{\lambda}}$).

For this model (``arithmetic random waves") it is known ~\cite{KKW
  2013,RW2014} that various local properties of $f_{n}$,
e.g. the length fluctuations of the nodal line $f_{n}^{-1}(0)$, the
number of nodal intersections against a reference curve, or the number
of nodal points with a given normal direction, depend on the limiting
angular distribution of $\{\lambda\in \Z^{2}:\|\lambda\|^{2}=n\}.$ For
example, in \cite{RW2008} the nodal length fluctuations for generic
eigenfunctions was shown to vanish (this can be viewed as a refinement
of Yau's conjecture \cite{yau1,yau2}), and in \cite{KKW 2013} the
leading order term of the variance of the fluctuations was shown to
depend on the angular distribution of
$\{\lambda\in \Z^{2}:\|\lambda\|^{2}=n\}$.
%More
%precisely, for $n\in S$ let
To make the notion of angular distribution precise, for $n\in S$ let
\begin{equation}
\label{eq:mun spect meas def}
\mu_{n} = \frac{1}{r_{2}(n)}\sum\limits_{\|\lambda\|^{2}=n}\delta_{\lambda/\sqrt{n}},
\end{equation}
where $\delta_{x}$ is the Dirac delta at $x$, be a probability measure
on the unit circle $\mathcal{S}^{1}\subseteq\R^{2}$.
It is then natural (or essential) to pass to subsequences $\{n_{j}\}\subseteq S$ such that
$\mu_{n_{j}}$ converges to some $\mu$ in the weak-$*$ topology, a probability measure on
$\mathcal{S}^{1}$, so that the various associated quantities,
such as the nodal length variance $\var(f_{n}^{-1}(0))$
exhibit an asymptotic law. In this situation we may identify $\mu$ as the
spectral density of the limiting field around each point of the torus
when the unit circle is considered embedded
$\mathcal{S}^{1}\subseteq\R^{2}$ (see Lemma \ref{lem:unif conv Fourier}); such a limiting probability measure
$\mu$ necessarily lies in the set $\Ptorus$ of ``monochromatic''
probability measures on $\mathcal{S}^{1}$, invariant w.r.t.
$\pi/2$-rotation and complex conjugation (i.e.
$(x_{1},x_{2})\mapsto (x_{1},-x_{2})$). In fact, the family of weak-*
partial limits of $\{\mu_{n}\}$ (``attainable'' measures) is known
\cite{KuWi attainable} to be a proper subset of $\Ptorus$.

\vspace{3mm}

Let $\Nc(f_{n})$ denote the total number of nodal components of
$f_{n}$ on $\Tb^{2}$. On one hand, an application of \cite[Theorem $4$]{So}
yields\footnote{Considering $c_{NS}$ in the more general sense as in
  Proposition \ref{prop:NS constant uniform},
and making the necessary adjustments
in case $\mu$ does not fall into the class of spectral measures considered by Nazarov-Sodin.}
that if, as above, $\mu_{n_{j}}\Rightarrow \mu$ for $\mu$ a
probability measure on $\mathcal{S}^{1}$, we have
\begin{equation}
\label{eq:E[N(fm)]=cNSm+o(m)}
\E[\Nc(f_{n_{j}})]= c_{NS}(\mu)\cdot n_{j} + o(n_{j}),
%\quad n_{j} \to \infty
\end{equation}
with the leading constant $c_{NS}(\mu)$ same as for the
scale-invariant model \eqref{eq:EN(f;R)=cNSR^m+o(Rm)}, cf. ~\cite[Theorem $1.2$]{Ro}. On the other hand,
we will be able to infer from Proposition \ref{prop:NS constant uniform} the more precise {\em uniform}
statement \eqref{eq:E[N]=cNS*n+O(sqrt(n))},
by considering $f_{n}$ on the square $[0,1]^{2}$ via the natural quotient map $q:\R^{2}\hookrightarrow \Tb^{2}$
(see the proof of Theorem \ref{thm:NS const torus} part \ref{it:N(fn)=c(mun)n+O(sqrt(n))}).

For $\mu \in \Ptorus$ we can classify all measures $\mu$
such that $c_{NS}(\mu)=0$, in particular classify when the leading constant on the r.h.s. of
\eqref{eq:E[N(fm)]=cNSm+o(m)} vanishes.
Namely, for $\theta\in [0,2\pi]$ let $$z(\theta):=(\cos(\theta),\sin(\theta))\in\Sc^{1}\subseteq\R^{2},$$
\begin{equation}
\label{eq:nu0 Cil meas}
\nu_{0} = \frac{1}{4}\sum\limits_{k=0}^{3}\delta_{z(k\cdot\pi/2)}
\end{equation}
be the {\em Cilleruelo} measure \cite{C 1993}, and
\begin{equation}
\label{eq:tildnu0 tilt Cil}
\widetilde{\nu_{0}}=\frac{1}{4}\sum\limits_{k=0}^{3}\delta_{z(\pi/4+k\cdot\pi/2)}
\end{equation}
be the {\em tilted} Cilleruelo
measure; these are the only measures in $\Ptorus$ supported
on precisely $4$ points. In addition to the aforementioned classification of measures $\mu\in\Ptorus$ with $c_{NS}(\mu)=0$
we prove the following concerning the
rate of convergence \eqref{eq:E[N(fm)]=cNSm+o(m)}, and the
range of possible constants $c_{NS}(\mu)$ appearing on the r.h.s. of
\eqref{eq:E[N(fm)]=cNSm+o(m)}. (Note that the Nazarov-Sodin constant on the
r.h.s. of \eqref{eq:E[N]=cNS*n+O(sqrt(n))} is associated with
$\mu_{n}$ as opposed to
the r.h.s. of \eqref{eq:E[N(fm)]=cNSm+o(m)}, which
is associated with the limiting measure $\mu$.)
%
% \new{(Mind that the Nazarov-Sodin constant on the
% r.h.s. of \eqref{eq:E[N]=cNS*n+O(sqrt(n))} is of $\mu_{n}$ as opposed to
% the limiting one $c_{NS}(\mu)$ on the r.h.s. of
% \eqref{eq:E[N(fm)]=cNSm+o(m)}.)}

\begin{theorem}
\label{thm:NS const torus}

\begin{enumerate}

\item
\label{it:N(fn)=c(mun)n+O(sqrt(n))}
We have uniformly for $n\in S$
\begin{equation}
\label{eq:E[N]=cNS*n+O(sqrt(n))}
\E[ \Nc(f_{n}) ] = c_{NS}(\mu_{n})\cdot n+O(\sqrt{n}),
\end{equation}
with the constant involved in the `$O$'-notation absolute.

\item
\label{it:|N(fn)-c(mun)n|=O(sqrt(n))}
If $\mu_{n_{j}}\Rightarrow\mu$ for some subsequence
$\{n_{j}\}\subseteq S$, where
$\mu$ has no atoms, then convergence in mean holds:
\begin{equation}
\label{eq:conv mean torus}
\E\left[\left| \frac{\Nc(f_{n_{j}})}{n_{j}} - c_{NS}(\mu)\right|\right]
=o_{\rho}(1).
\end{equation}

\item
\label{it:cNS(rho)=0 <=> Cil}
For $\mu\in \Ptorus$,
$c_{NS}(\mu)=0$ if and only if $\mu=\nu_{0}$ or $\mu = \widetilde{\nu_{0}}$.

\item
\label{it:cNS(Psymm)=[0,dmax]}
For $\mu$ in the family of weak-* partial limits of
$\{\mu_{n}\}$, the functional $c_{NS}(\mu)$ attains all values in an interval of the form
$I_{NS}=[0,d_{max}]$ with some $d_{max}>0$.

\end{enumerate}

\end{theorem}

It is opportune to mention that D. Beliaev has informed us that he, together with
M. McAuley and S. Muirhead, recently obtained a full classification the
set of measures $\rho\in\Pc$ for which $c_{NS}(\rho)=0$.

\subsection{Acknowledgments}

It is a pleasure to thank M. Sodin for
many stimulating and fruitful discussions, insightful and critical
comments while conducting the research presented, and his comments on an earlier
version of this manuscript. We would also like
to thank Z. Rudnick for many fruitful discussions and his help in
improving the present manuscript, D. Beliaev for his valuable comments, in particular regarding the
Nazarov-Sodin discrepancy functional $d_{NS}(\cdot)$, J. Buckley and M. Krishnapur for many stimulating
conversations, and S. Muirhead for pointing the Gaussian
Correlation Inequality ~\cite{Roy} to us, as well as other useful
comments on an earlier version of this manuscript.
Finally, we thank P. Sarnak and J. Bourgain for their interest
in our work and their support, and the anonymous referee for his valuable
comments that improved the readability of our manuscript.

P.K. was partially supported by grants from the G\"oran Gustafsson
Foundation, and the Swedish Research Council (621-2011-5498 and
621-2016-03701.)
I.W. was partially supported by the EPSRC under the First Grant scheme
(EP/J004529/1). Further, the research leading to these results has received funding from the
European Research Council under the European Union's Seventh
Framework Programme (FP7/2007-2013) / ERC grant agreement
n$^{\text{o}}$ 335141 (I.W.).

%\section{Outline of the paper}

\section{Discussion and outline of key ideas}
\label{sec:discussion}

\subsection{Continuity of the number of nodal domains}

Theorem \ref{thm:cNS cont}, a principal result of this paper, states that the expected number
$\E\left[\Nc(f_{\rho};R)\right]$ of nodal domains of $f_{\rho}$ lying in a compact domain of $\R^{2}$,
properly normalized, is continuous in the limit $R\rightarrow\infty$, namely
$c_{NS}(\rho)$. We believe that it is in fact continuous without taking the limit, i.e. for $R$ {\em fixed},
the function
$$
\rho \to \E[\Nc(f_{\rho};R)]
$$
is a continuous function on $\Pc$.

\subsection{Maximal Nazarov-Sodin constant}

As for the maximal possible value of $c_{NS}$, it seems {reasonable}
to assume that, in order to maximize the nodal domains number for
$\rho\in\Pc$, one had better maximize the weight of the highest
possible wavenumber. That is, to attain $c_{\max}$ as in Corollary
\ref{cor:NS constant vary Eucl} the measure $\rho$ should be supported
on $\Sc^{1}\subseteq \R^{2}$, i.e. the random wave $f_{\rho}$ must be
monochromatic. Among those measures $\rho\in\Pc$ supported in
$\Sc^{1}$ we know that the more concentrated ones (i.e. those
supported on two antipodal points, or, for $\rho\in\Ptorus$,
Cilleruelo measure $\nu_{0}$ supported on $4$ symmetric points
$\pm 1,\pm i$) {\em minimize} the nodal domains number (Theorem
\ref{thm:NS const torus}, part \ref{it:cNS(rho)=0 <=> Cil}); (tilted)
Cilleruelo measure is known to minimize other local quantities
~\cite{KKW 2013} when the uniform measure maximizes it, or vice
versa.

For example, it is easy to see that $\E[\Nc(f,R)] $ is bounded
above by the expected number of points $x=(x_{1},x_{2}) \in \Dc_{R}$
such that
$$f(x) = 0 = \frac{\partial f}{\partial x_{1}} (x) + \frac{\partial f}{\partial x_{2}} (x),$$
and this expectation can be shown to be minimal
for the Cilleruello measure. Now, as the upper bound expectation is
not invariant with respect to change of coordinates via rotation it is
natural to chose the optimal rotation. That is, given a spectral
measure $\rho$ one should optimize by choosing a rotation that
minimizes the above upper bound. The Cilleruello measure, as well as
the twisted one, has a minimal optimized upper bound, whereas the
uniform measure has a maximal optimized upper bound.

It thus seems plausible that
the uniform measure $\rho=\frac{d\theta}{2\pi}$ on $\Sc^{1}$ corresponding to the Nazarov-Sodin constant $c_{RWM}$ in \eqref{eq:cRWM>0}
{\em maximizes} the Nazarov-Sodin constant; since it happens $\frac{d\theta}{2\pi}\in\Ptorus$ to lie in $\Ptorus$, and is also
a weak-$*$ limit of $\{\mu_{n}\}$ in \eqref{eq:mun spect meas def},
it then also maximizes the Nazarov-Sodin constant restricted as in Theorem \ref{thm:NS const torus}.
The above discussion is our motivation for the following conjecture regarding the maximal possible values $c_{\max}$ (resp. $d_{\max}$)
of the Nazarov-Sodin constant.

\begin{conjecture}
\label{conj:max-ns-constant}

\begin{enumerate}

\item For $\mu\in \Ptorus$ that are weak-$*$ limits of $\{\mu_{n}\}$,
the maximal value $d_{\max}$ is uniquely
attained by $c_{NS}(\mu_{\mathcal{S}^{1}})$, where $\mu_{\mathcal{S}^{1}}$
is the uniform measure on $\mathcal{S}^{1}\subseteq\R^{2}$. In particular,
$$d_{\max}=c_{\text{RWM}}.$$

\item For $\rho\in \Pc$, the maximal value $c_{\max}$ is uniquely
attained by $c_{NS}(\rho)$ for $\rho$ the
uniform measure on $\mathcal{S}^{1}\subseteq\R^{2}$. In particular
$$c_{\max}=d_{\max}=c_{\text{RWM}}.$$

\end{enumerate}

\end{conjecture}

\subsection{Cilleruelo sequences for arithmetic random waves}

On one hand Theorem \ref{thm:NS const torus} shows that, if one stays away from the Cilleruelo measure, it is possible to
infer the asymptotic behaviour of the toral nodal domains number $\Nc(f_{n})$ from the asymptotic behaviour
of $\Nc(f_{\rho};\cdot)$ where $\rho=\mu_{n}$ is the spectral measure of $f_{n}$ when considered on $\R^{2}$.
On the other hand, if $\{n_{j}\}\subseteq S$ is a Cilleruelo sequence, i.e.,
\begin{equation}
\label{eq:munj=>Cil}
\mu_{n_{j}}\Rightarrow \nu_{0},
\end{equation}
then from part \ref{it:cNS(rho)=0 <=> Cil} of Theorem \ref{thm:NS const torus}
we can only infer that
$$\lim_{j \to \infty} \frac{\E[\Nc(f_{n_{j}})]}{n_{j}}  =  0,$$
with no further understanding of the true asymptotic behaviour of
$\E[\Nc(f_{n_{j}})]$.

It is possible to realize the Euclidean random field
$f_{\nu_{0}}:\R^{2}\rightarrow\R$ as a trigonometric polynomial (for
more details, see \eqref{eq:frho Cil 4 sum} or \eqref{eq:f0 Cil def}),
with only $4$ nonzero coefficients (see the $1$st proof of Lemma
\ref{lem:cNS(Cil)=0} below); a typical sample of the corresponding
nodal pictures are shown in Figure~\ref{fig:solCilhorvert} (cf. \S~\ref{sec:NS(Cil)=0 proof}.)
We may deduce that
a.s. $\Nc(f_{\nu_{0}};\cdot)\equiv 0$, i.e. there are no compact
domains of $f_{\nu_{0}}$ at all and all the domains are either
predominantly horizontal or predominantly vertical, occurring with
probability $\frac{1}{2}$.  The analogous situation on the torus
occurs for $n=m^{2}$ with
\begin{equation*}
g_{0;m}(x) = \frac{1}{\sqrt{2}} \cdot \left(a_{1}\cdot \cos(m\cdot x_{1}+\eta_{1}) + a_{2}\cdot \cos(m\cdot x_{2}+\eta_{2})   \right),
\end{equation*}
where $a_{1}$, $a_{2}$ are $\text{Rayleigh}(1)$
distributed independent random variables, and $\eta_{1},\eta_{2}\in [0,2\pi)$ are random phases uniformly drawn in $[0,2\pi)$;
in this case the nodal components in Figure \ref{fig:solCilhorvert} all become periodic with nontrivial homology, and
their number is of order of magnitude
\begin{equation}
\label{eq:Nc Cil approx sqrt(n)}
\Nc(g_{0}) \approx m \approx \sqrt{n}.
\end{equation}

Since the Nazarov-Sodin constant does not vanish outside of the
(tilted) Cilleruelo measure,
for every finite instance $f_{n}$ with $n\in S$
one would expect for more domains as compared to
\eqref{eq:Nc Cil approx sqrt(n)}, whether or not $n$ is a square, i.e.,
$\Nc(f_{n})\gg\sqrt{n}$. The above intuition has some reservations. A fragment of a sample nodal portrait of
$f_{n}$ with $n$ corresponding to a measure $\mu_{n}$ close to
Cilleruelo is given in Figure \ref{fig:cilleruello}.
\begin{figure}[ht]
\centering
\includegraphics[width=70mm]{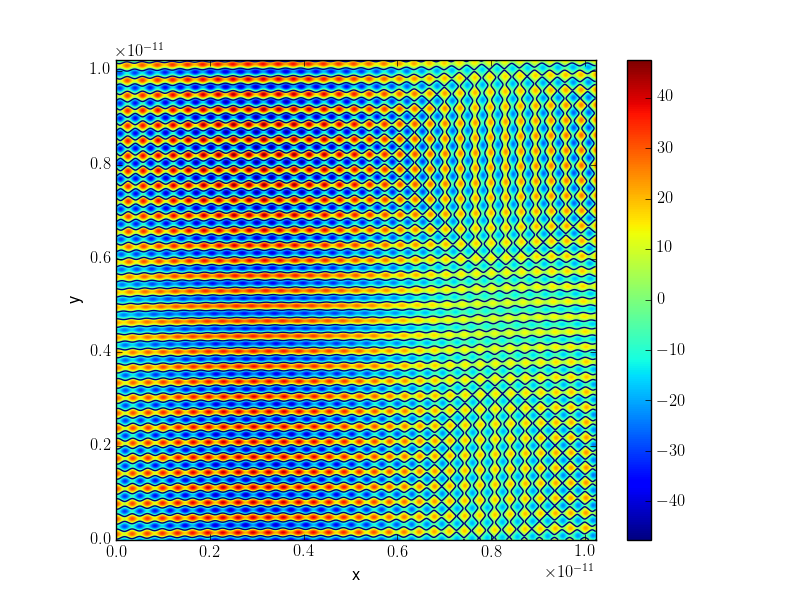}
\includegraphics[width=50mm]{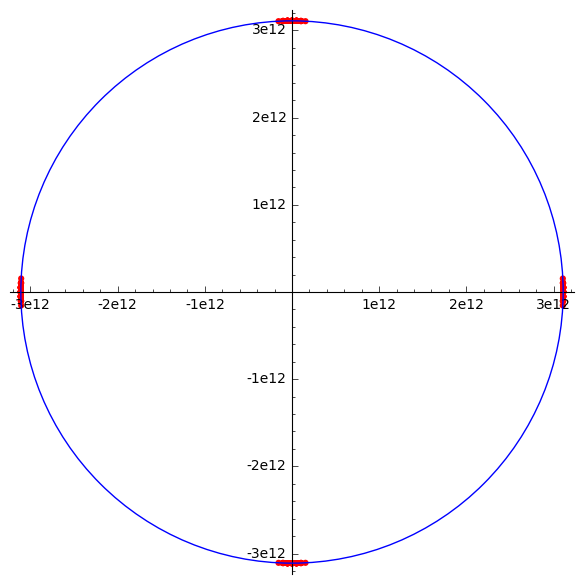}
\caption{Left: Plot of a fragment of a random ``Cilleruello'' type
  eigenfunction, nodal curves in black.  Right: corresponding spectral
  measure.  Here $n=9676418088513347624474653$ and $r_{2}(n) = |\{
  (x,y) \in \Z^2 : x^{2}+y^{2} = n \}| = 256$; for this particular choice of $n$, the corresponding
  lattice points, shown in red, are concentrated around $4$ antipodal points.}
\label{fig:cilleruello}
\end{figure}

It exhibits that, just as in Figure \ref{fig:solCilhorvert}, the nodal
domains are all predominantly horizontal or vertical, but the
suggested effect of the perturbed Cilleruelo shows that the periodic
trajectories sometime connect in some percolation-like process, and
transform from horizontal to vertical and back. Judging from the small
presented fragment it
%virtually impossible to guess whether this procedure decreases the
seems difficult to determine
to what extent this procedure decreases the total number of nodal
domains, in particular whether the expectation is bounded or not.
%
%whether this procedure decreases the
%total number of nodal domains sufficiently to make the expectation
%bounded.
For a higher resolution picture, as well as some further
examples of Cilluello eigenfunctions, see Appendix~\ref{sec:cill-plots}.

In any case it is likely that the genuine asymptotic
behaviour of $\E[\Nc(f_{n_{j}})]$ depends on the rate of convergence
\eqref{eq:munj=>Cil}, hence does not admit a simple asymptotic law.
With all our reservations, the above discussion is our basis for the
following question.

\begin{question}
Is it true that for an arbitrary Cilleruelo sequence,
$$  \liminf\limits_{j\rightarrow\infty} \E[ N(f_{n_{j}}) ]\rightarrow \infty, $$
or, even stronger
$$   \E[ N(f_{n_{j}}) ]\gg \sqrt{n_{j}}?$$

\end{question}

If, as we tend to think, the answer to the latter question is ``yes",
then a simple compactness argument yields that for the full
sequence $n\in S$ we have
$$\lim\limits_{n\rightarrow\infty} \E[\Nc(f_{n})] = \infty.$$

\subsection{The true nature of the Nazarov-Sodin constant}

Motivated by the fact that various local quantities, such as
the nodal length variance ~\cite{KKW 2013}, or the expected number of ``flips"
(see \eqref{eq:S1,S2 def}) or critical points, only depends on the first
non-trivial Fourier coefficient of the measure, we raise the following
question.

\begin{question}
Is it true that $c_{NS}(\mu)$ with $\mu\in\Ptorus$
only depends on finitely many Fourier
coefficients, e.g. $\widehat{\mu}(4)$ or
$(\widehat{\mu}(4),\widehat{\mu}(8))$?
\end{question}

\subsection{Key ideas of the proof of Theorem \ref{thm:cNS cont}}
% \begin{verbatim}
% Say the independent interest, rate, and some ideas where it could
% help.  (Now done in intro.)
% \end{verbatim}
%

To prove the continuity of $c_{NS}$ in Theorem \ref{thm:cNS cont} we wish to show that
$|c_{NS}(\rho)-c_{NS}(\rho')|$ is small for two ``close'' spectral
measures $\rho, \rho'$. To this end we show that for a large $R$ there exists a
coupling between the random fields $f_{\rho}$ and $f_{\rho'}$, and
that $\Nc(f_{\rho},R)$ and $\Nc(f_{\rho'},R)$ are very likely to be
close (in fact, that $f_{\rho}$ and $f_{\rho'}$ are $C^{1}$-close and
that they have essentially the same nodal components). Of key
importance is that both $f_{\rho}$ and $f_{\rho'}$ are not only
$C^{1}$-close, but also likely to be ``stable'' in the sense that
small perturbations do not change the number of nodal components,
except near the boundary. However, %given $\rho,\rho'$
we can only prove stability, and the desired properties of the
coupling, for square domains $\Dc_{R}$ for $R$ {\em fixed}, and it is
thus {essential} to have bounds on the difference
$$\Nc(f_\rho,R)/(4R^{2}) - c_{NS}(\rho)$$ that are {\em uniform} in both
$\rho$ and $R$.

To obtain uniformity in $R$ we tile a ``huge'' square with a fixed
``large'' square, and count nodal domains entirely contained in the
fixed large square.  By translation invariance, the expectation over
all large squares is the same, hence the scaled number of components
in the large square (i.e., scaling by dividing by the area of the
square) is the same as the scaled number of components of the huge
square, up to an error involving the (scaled) number of nodal
components that intersect a boundaries of at least one large square.
This in turn can be {\em uniformly bounded} (in terms of $\rho$) by
using Kac-Rice type techniques to uniformly bound the expected number
of zeros of $f_{\rho}$ lying on a curve (the bound of course depends
on its length), cf. Lemma \ref{lem:Kac Rice intersections}.

In case the huge square cannot exactly be tiled by large squares, we
make use of the following observation: the number of nodal domains
entirely contained in some region is bounded from above by the number
of ``flip points'', i.e., points $x=(x_{1},x_{2})$ where
$f = \frac{\partial f}{\partial x_{1}} = 0$, and the expected number of such
points is, up to a uniform constant, bounded by the area of the region.
To show this we again use a Kac-Rice type ``local'' estimates,
cf. Lemma~\ref{lem:Kac Rice flips} and its proof.

\vspace{2mm}

Nazarov and Sodin assume
that the support of $\rho$ is not contained in a line, in order for
non-degeneracy of $( f_{\rho}, \nabla f_{\rho})$ to hold.  Now, if
$\rho_{i} \Rightarrow \rho$ and the limiting measure $\rho$ is
non-degenerate, there exists $\epsilon>0$ such that $\rho$ is outside
a small neighbourhood $\Peps$ of the degenerate measures within $\Pc$
defined in \S~\ref{sec:Kac-Rice} below (see \ref{eq:Pesp def}).
The outlined approach above yields continuity of $c_{NS}$ around
$\rho$ in the complement $\Pc \setminus \Peps$.

On the other hand, if the limit $\rho$ is degenerate we use a separate
argument. First we show that $c_{NS}(\rho) = 0$ by showing that
$f_{\rho}$, almost surely, has no bounded nodal domains; similarly
this shows that we may assume that all $\rho_i$ gives rise to
non-degenerate fields. As non-degeneracy holds along the full
sequence, we can then use Kac-Rice type local argument giving that
$\E(\Nc(f_{\rho_{i}};R))/R^{2} \to 0$ as $i \to \infty$.

\section{Kac-Rice premise}
\label{sec:Kac-Rice}

%\subsection{Notational conventions}
%\label{sec:notat-conv}
%For the convenience of the reader
We begin with collecting some
notational conventions.
Given a smooth function $f$ on $\R^2$
let $f_{1} = \partial_1 f = \frac{\partial f}{\partial x_{1}} $,
$f_{2} = \frac{\partial f}{\partial x_{2}}$, and
$f_{12} = \partial_1 \partial_2 f = \frac{\partial^{2} f}{\partial x_{1} \partial x_{2}} $ (etc),
where $x = (x_{1},x_{2}) \in \R^2$; and similarly for smooth functions $f : \T^2 \to \R$.

The Kac-Rice formula is a standard tool for computing moments of
various {\em local} properties of random (Gaussian) fields, such
as, for example, number of nodal intersections against a reference curve, number of critical points etc. For our purposes
we will not require any result beyond the expectation of the number of
zeros $\Zc_{F}$ of a stationary Gaussian field
$F:\Dc\rightarrow\R^{n}$
on a compact domain $\Dc\subseteq \R^{n}$ (closed interval for $n=1$),
with the sole intention of applying it in the $2d$ case. For $x\in \Dc$ define the zero density as the conditional Gaussian expectation
\begin{equation}
\label{eq:K1 density def}
K_{1}(x) = K_{1;F}(x)= \phi_{F(x)}(0)\cdot\E[|\det{J_{F}(x)}|\big| F(x)=0],
\end{equation}
where $\phi_{F(x)}$ is the probability density function of the Gaussian vector $F(x)\in \R^{n}$, and $J_{F}(x)$ is the
Jacobian matrix of $F$ at $x$. The Kac-Rice meta-theorem states that, under some non-degeneracy conditions on $F$,
\begin{equation}
\label{eq:E[ZF]=intK(1) gen}
\E[\Zc_{F}] = \int\limits_{\Dc}K_{1}(x)dx;
\end{equation}
to our best knowledge the mildest sufficient conditions for the validity of \eqref{eq:E[ZF]=intK(1) gen},
due to Azais-Wschebor ~\cite[Theorem 6.3]{AW}, is that
for all $x\in \Dc$ the distribution of the Gaussian vector $F(x)\in \R^{n}$ is non-degenerate.

As for the zero density $K_{1}(x)$ in \eqref{eq:K1 density def}, since \eqref{eq:K1 density def} is a Gaussian expectation
depending on the law of $(F(x),J_{F}(x))$, it is in principle possible
to express $K_{1}$ in terms
of the
covariance of $F$.
For $F$ a derived field of $f_{\rho}$ (if, for example, $F$ is a restriction of $f_{\rho}$ on the reference curve $\Cc$ in Lemma
\ref{lem:Kac Rice intersections} below, or $F=(f_{\rho},\partial_{1}f_{\rho})$ in Lemma \ref{lem:Kac Rice flips} below) it is possible to express $K_{1}(x)$ in terms of
the covariance function $$r_{\rho}(x,y)=\E[f_{\rho}(x)\cdot f_{\rho}(y)]$$ and its various derivatives, or,
what is equivalent, its spectral measure $\rho$,
supported on $B(1)$:
\begin{equation}
\label{eq:r_rho Fourier}
r_{\rho}(x) = \int\limits_{B(1)} e(\langle x, y \rangle) d\rho(y),
\end{equation}
where $$e(t)=e^{2\pi i t}.$$
In case $F$ is stationary (see Lemma \ref{lem:Kac Rice flips} below), $K_{1;F}(x)$ in \eqref{eq:K1 density def} is independent
of $x$; in this case to prove a uniform upper bound we only need to control it in terms of $F$.

For the Kac-Rice method to apply it is essential that the field is
nondegenerate. In order to analyze certain degenerate limit measures
we introduce the following notation.
Given a stationary Gaussian field $f_{\rho}$ with spectral measure $\rho$, let
$C(\rho)$ denote the positive semi-definite covariance matrix
\begin{equation}
\label{eq:Crho def}
C(\rho) :=
\begin{pmatrix}
\var(\partial_1 f_{\rho}(0)) & \cov(\partial_1 f_{\rho}(0),\partial_2 f_{\rho}(0))\\
\cov(\partial_1 f_{\rho}(0),\partial_2 f_{\rho}(0)) & \var(\partial_2 f_{\rho}(0)) \\
\end{pmatrix},
\end{equation}
and let $\lambda(\rho) \geq 0$ denote the smallest eigenvalue of
$C(\rho)$. As the map $\rho \to C(\rho)$ is continuous, the same
holds for $\rho \to \lambda(\rho)$.
Thus, if we are given $\epsilon >0$ define
\begin{equation}
\label{eq:Pesp def}
\Peps := \{ \rho \in \Pc : \lambda(\rho) < \epsilon\}
\end{equation}
we find that $\Pc \setminus \Peps$ is a closed subset of $\Pc$.
Abusing notation slightly it is convenient to let
$$
\Pzero := \{ \rho \in \Pc : \lambda(\rho) =0 \}
$$
denote the set spectral measures giving rise to degenerate fields.
(We may interpret the covariance matrix $C(\rho)$ as a matrix
representing a positive semi-definite quadratic form; non-degenaracy
is then equivalent to the form being positive definite.
As quadratic forms in
two variables can be diagonalised by a rotation, degeneracy implies that
after a change of coordinates by rotation, we have $\var( \partial_{1}
f) = \int \xi_1^{2} \, d\rho(\xi) =  0$, and hence the support of
$\rho$ must be contained in the  line $\xi_1=0$.)

\vspace{3mm}

As we intend to apply the Kac-Rice formula on $f_{\rho}$, for
$\rho\in\Pc \setminus \Peps$ for some $\epsilon>0$, and some derived
random fields (see lemmas \ref{lem:Kac Rice intersections} and
\ref{lem:Kac Rice flips} below) we will need to collect the following
facts.

\begin{lemma}
\label{lem:covar der basic}
\begin{enumerate}

\item For every unit variance random field $F:\Dc\rightarrow\R$, and
  $x\in\Dc$, the value $F(x)$ is independent of the gradient $\nabla
  F(x)$.

\item The variances
  $\var(\partial_{1} f_{\rho}), \var(\partial_{2} f_{\rho})$ of the
  first partial derivatives is bounded away from $0$, uniformly for
  $\rho\in \Pc \setminus \Peps$.

\end{enumerate}

\end{lemma}

The proof of Lemma \ref{lem:covar der basic} will be given in
\S~\ref{sec:loc proofs}.

\subsection{Intersections with curves and flips}

We begin with a bound on expected number of nodal intersections with
curves, whose proof will be given in \S~\ref{sec:loc proofs}.
\begin{lemma}
\label{lem:Kac Rice intersections}
Let $\Cc\subseteq \R^{2}$ be a smooth curve of length $\Lc$, and
$\Nc(f_{\rho},\Cc)$ the number of nodal intersections
of $f_{\rho}$ with $\Cc$, $\rho\in\Pc$. Then
\begin{equation*}
\E[\Nc(f_{\rho},\Cc)] = O(\Lc)
\end{equation*}
with constant involved in the `$O$'-notation {\em absolute}.
\end{lemma}

The notion of ``nodal flips'' will be very useful for giving uniform upper
bounds on the number of nodal domains.
\begin{notation}
  For $\Dc\subseteq \R^{2}$ a nice closed
  domain we denote the number of vertical and horizontal nodal flips
\begin{equation}
\label{eq:S1,S2 def}
\begin{split}
S_{1}(f_{\rho};\Dc) &= \#\{x\in \Dc :\: f_{\rho}(x)= \partial_{1} f_{\rho}(x) = 0\}, \\
S_{2}(f_{\rho};\Dc) &= \#\{x\in \Dc :\: f_{\rho}(x)= \partial_{2} f_{\rho}(x) = 0\},
\end{split}
\end{equation}
respectively.

\end{notation}

%\begin{lemma}[Bound for expected number of flips]
\begin{lemma}
\label{lem:Kac Rice flips}
%We have the following bound on the minimum of the
%expected number of horizontal and vertical nodal
%flips:
For all $\rho\in\Pc\setminus \Pc_{0}$, we have
\begin{equation}
  \label{eq:flip-bound-via-variance}
  \begin{split}
    \E[S_{1}(f_{\rho};\Dc)]
& = O\left(\area(\Dc) \cdot
\var(\partial_{2} f_{\rho})^{1/2}\right),
\\
\E[S_{2}(f_{\rho};\Dc)]
& = \left(\area(\Dc) \cdot
\var(\partial_{1} f_{\rho})^{1/2}\right),
  \end{split}
\end{equation}
and consequently
\begin{equation*}
\max( \E[S_{1}(f_{\rho};\Dc)], \E[S_{2}(f_{\rho};\Dc)]) = O(\area(\Dc))
\end{equation*}
with constants involved in the `$O$'-notation {\em absolute}.
\end{lemma}

Lemma \ref{lem:Kac Rice flips} will
be proved in \S~\ref{sec:loc proofs}. As it was mentioned above, for
$\rho \in \Peps$ we may
arrange that, after rotating if necessary, either
$\var(\partial_{1} f_{\rho})$ or $\var(\partial_{2} f_{\rho})$ is
at most
$\epsilon$. To treat the degenerate case
$\rho\in \Pc_{0}$ we record the following fact.

\begin{lemma}
\label{lem:degenerate}
If $\rho \in \Pzero$ then $\Nc(f_{\rho};\cdot)\equiv 0$. In particular in this case
\eqref{eq:N(R)=cR^2+O(R)} holds with $c_{NS}(\rho) = 0$.
\end{lemma}

\begin{proof}
After changing coordinates by a rotation, we may assume that
$\var( \partial_{1} f) = 0$. Hence, almost surely, we have
$f(x_{1},x_{2}) = g(x_{2})$ for some function $g$, and thus $f$ has
no compact nodal domains, and in particular $c_{NS}(\rho)=0$.
\end{proof}

\subsection{Proof of Proposition \ref{prop:NS constant uniform}}

\begin{proof}

First, we may assume that $\rho\in\Pc\setminus\Pc_{0}$ by the virtue
of Lemma \ref{lem:degenerate}, so that we are eligible to apply
Lemma \ref{lem:Kac Rice flips} on $f_{\rho}$.  Now
let $R_{2}\ggg R_{1}\ggg 0$ be a two big real numbers; for notational
  convenience we will at first assume that
\begin{equation}
\label{eq:R2=kR1}
R_{2}=kR_{1}
\end{equation}
is an integer multiple of $R_{2}$, $k\gg 1$. We divide the square $\Dc_{R_{2}}=[-R_{2},R_{2}]^{2}$ into
$4k^{2}=4\frac{R_{2}^{2}}{R_{1}^{2}}$ smaller squares $$\{\Dc_{R_{2};i,j}\}_{ i,j =1\ldots 2k}$$ of side length $R_{1}$,
disjoint save to boundary overlaps.
Every nodal component lying in $\Dc_{R_{2}}$ is either lying entirely
in one of the $\Dc_{R_{2};i,j}$ or intersects at least one of
the vertical or horizontal line segments, $\{x= i\cdot T,\,  |y|\le
R_{2}\}$, $i=-k,\ldots k$, or
$\{y= j\cdot T,\, |x|\le R_{2}\}$, $j=-k,\ldots,k$ respectively.
Let $\Nc(f_{\rho};\Dc_{R_{2};i,j})$ be the number of nodal components of $f_{\rho}$ lying in $\Dc_{R_{2};i,j}$, and
$\Zc(f_{\rho}; R_{2},x=iR_{1})$, $\Zc(f_{\rho}; R_{2},y=jR_{1})$ be the number of nodal intersections of $f_{\rho}$ against a finite
vertical or horizontal line segment as above.

The above approach shows that
\begin{equation}
\label{eq:numb comp=lying small + nod int}
\begin{split}
\Nc(f_{\rho};R_{2}) &= \sum\limits_{1 \le i,j \le 2k}\Nc(f_{\rho};\Dc_{R_{2};i,j}) +
O\left( \sum\limits_{i=-k}^{k} \Zc(f_{\rho}; R_{2},x=iR_{1})\right) \\&+
O\left( \sum\limits_{j=-k}^{k} \Zc(f_{\rho}; R_{2},y=jR_{1})\right).
\end{split}
\end{equation}
We now take expectation of both sides of \eqref{eq:numb comp=lying
  small + nod int}; using the translation invariance of $f_{\rho}$,
and
Lemma \ref{lem:Kac Rice intersections} we find that
\begin{equation}
\label{eq:E[N(R2)]=4k^2EN(R1/2)+O(kR2)}
\E[\Nc(f_{\rho};R_{2})] = 4k^{2}\cdot \E[\Nc(f_{\rho};R_{1}/2)] + O(kR_{2}),
\end{equation}
where the constant involved in the `$O$'-notation is absolute. A
simple manipulation with \eqref{eq:E[N(R2)]=4k^2EN(R1/2)+O(kR2)},
bearing
in mind \eqref{eq:R2=kR1}, now implies
\begin{equation}
\label{eq:N(R)/R^2 Cauchy}
\left|\frac{\E[\Nc(f_{\rho};R_{2})]}{4R_{2}^{2}} -   \frac{\E[\Nc(f_{\rho};R_{1}/2)]}{R_{1}^{2}} \right| = O\left(\frac{1}{R_{1}} \right)
\end{equation}
with the constant involved in the `$O$'-notation absolute, with \eqref{eq:N(R)/R^2 Cauchy}.

In case $R_{2}$ is not an integer multiple of $R_{1}$, in the above argument we leave a small rectangular corridor of size at most
$R_{1}\times R_{2}$ (in fact, two such corridors). In this case the estimate \eqref{eq:N(R)/R^2 Cauchy} should be replaced
by
\begin{equation}
\label{eq:N(R)/R^2 Cauchy remainder}
\left|\frac{\E[\Nc(f_{\rho};R_{2})]}{4R_{2}^{2}} -   \frac{\E[\Nc(f_{\rho};R_{1}/2)]}{R_{1}^{2}} \right| = O\left(\frac{1}{R_{1}}+ \frac{R_{1}}{R_{2}} \right),
\end{equation}
with $O(\frac{R_{1}}{R_{2}})$ coming from the contribution of the small rectangular leftover corridor thinking of $R_{2}$ much bigger
than $R_{1}$; here we used Lemma \ref{lem:Kac Rice flips}, valid since we assumed $\rho\in\Pc\setminus\Pc_{0}$.
The latter estimate \eqref{eq:N(R)/R^2 Cauchy remainder} shows that $\left\{\frac{\E[\Nc(f_{\rho};R)]}{4R^{2}}\right\}$
satisfies the Cauchy convergence criterion (if $R_{1}$ and $R_{2}$ are of comparable size then we use the triangle inequality,
after tiling both $\Dc_{R_{2}}$ and $\Dc_{R_{1}/2}$ with
much finer mesh size),
we then denote its limit
by
$$
c_{NS}(\rho) :=
\lim\limits_{R\rightarrow\infty}\frac{\E[\Nc(f_{\rho};R)]}{4R^{2}}.
$$

\vspace{3mm}

Now that the existence of the limit $c_{NS}(\rho)$ is proved, we may
assume that $R_{2}$ is an integer multiple of $R_{1}$, and take the
limit $R_{2} \rightarrow\infty$ in \eqref{eq:N(R)/R^2 Cauchy}; it
yields
\begin{equation}
\label{eq:|EN(R1)/R1^2-cNS|<<1/R^2}
\left| \frac{\E[\Nc(f_{\rho};R_{1}/2)]}{R_{1}^{2}} - c_{NS}(\rho) \right| = O\left(\frac{1}{R_{1}} \right),
\end{equation}
again with the constant in the`$O$'-notation on the
r.h.s. absolute. We conclude the proof of Proposition \ref{prop:NS constant uniform}
by noticing that \eqref{eq:|EN(R1)/R1^2-cNS|<<1/R^2} is a restatement
of \eqref{eq:N(R)=cR^2+O(R)} (e.g. replace $R_{1}$ by $R/2$).

\end{proof}

\subsection{Proof of Proposition \ref{prop:NS constant uniform part two}}

\begin{proof}

Again, for $\rho\in\Pc_{0}$ there is nothing to prove here thanks to Lemma \ref{lem:degenerate}, so that from this point on
we assume that $\rho\in\Pc\setminus\Pc_{0}$. Let
\begin{equation}
\label{eq:xi=liminf def}
\xi=\xi(\rho) := \liminf\limits_{R\rightarrow\infty}\E\left[\left| \frac{\Nc(f_{\rho};R)}{4R^{2}}-c_{NS}(\rho) \right|\right];
\end{equation}
in what follows we argue that, in fact, $\xi$ is the limit. Let $\epsilon>0$ be given and $R_{1}=R_{1}(\rho,\epsilon)>0$
such that
\begin{equation}
\label{eq:N(R1)/4R^2-c<xi+eps}
\E\left[\left| \frac{\Nc(f_{\rho};R_{1}/2)}{R_{1}^{2}} -c_{NS}(\rho)\right|\right] < \xi + \epsilon.
\end{equation}

Following the proof of Proposition \ref{prop:NS constant uniform} let $R_{2}\gg R_{1}\gg 0$ be a large real number;
as before we divide the square $\Dc_{R_{2}}=[-R_{2},R_{2}]^{2}$ into the smaller squares $\{ \Dc_{R_{2},i,j}\}_{1\le i,j \le 2k} $ of side length
$R_{1}$ leaving a couple of corridors of size at most $R_{1}\times R_{2}$, and write (cf. \ref{eq:numb comp=lying small + nod int})
\begin{equation}
\label{eq:numb comp=lying small + nod int corr}
\begin{split}
0&\le \Nc(f_{\rho};R_{2})- \sum\limits_{1 \le i,j \le 2k}\Nc(f_{\rho};\Dc_{R_{2};i,j}) \le
\sum\limits_{i=-k}^{k} \Zc(f_{\rho}; R_{2},x=iR_{1})\\&+
\sum\limits_{j=-k}^{k} \Zc(f_{\rho}; R_{2},y=jR_{1}) + \Nc(f_{\rho};\Fc_{R_{2},R_{1}}),
\end{split}
\end{equation}
where we denoted $\Fc_{R_{2},R_{1}}$ to be the union of the two leftover rectangular corridors, and
$\Nc(f_{\rho},\Fc_{R_{2},R_{1}})$ the corresponding number of nodal domains lying entirely inside $\Fc_{R_{2},R_{1}}$.

Taking the expectation of both sides of \eqref{eq:numb comp=lying small + nod int corr} and dividing by $4R_{2}^{2}$ we have that
(using the non-negativity of the l.h.s. of \eqref{eq:numb comp=lying small + nod int corr})
\begin{equation}
\label{eq:E[|N(R2)/4R2^2-sum|]<<1/R1+R1/R2}
\begin{split}
\E\left[\left|\frac{\Nc(f_{\rho};R_{2})}{4R_{2}^{2}}- \frac{1}{4k^{2}}\cdot \sum\limits_{1 \le i,j \le 2k}\frac{\Nc(f_{\rho};\Dc_{R_{2};i,j})}{R_{1}^{2}}\right|\right]=
O\left(\frac{1}{R_{1}}+ \frac{R_{1}}{R_{2}} \right),
\end{split}
\end{equation}
thanks to Lemma \ref{lem:Kac Rice flips}, valid for $\rho\in\Pc\setminus\Pc_{0}$. On the other hand, by \eqref{eq:N(R1)/4R^2-c<xi+eps},
the triangle inequality, and the translation invariance of $f_{\rho}$, we have that
\begin{equation}
\label{eq:E[sum-cNS] sm triangle}
\begin{split}
&\E\left[\left|\frac{1}{4k^{2}}\cdot \sum\limits_{1 \le i,j \le 2k}\frac{\Nc(f_{\rho};\Dc_{R_{2};i,j})}{R_{1}^{2}}-c_{NS}(\rho)\right|\right]
\\&\le
\frac{1}{4k^{2}}\sum\limits_{1 \le i,j \le 2k} \E\left[\left| \frac{\Nc(f_{\rho};\Dc_{R_{2};i,j})}{R_{1}^{2}} - c_{NS}(\rho) \right|\right]
\\&= \E\left[ \left| \frac{\Nc(f_{\rho};R_{1}/2)}{R_{1}^{2}} - c_{NS}(\rho) \right|\right] < \xi+\epsilon.
\end{split}
\end{equation}
We have then
\begin{equation}
\label{eq:E[|N2/4R2^2-cNS|]<1/R1+R1/R2+xi+eps}
\begin{split}
&\E\left[\left|\frac{\Nc(f_{\rho};R_{2})}{4R_{2}^{2}} - c_{NS}(\rho)\right|\right] \le
\E\left[\left|\frac{\Nc(f_{\rho};R_{2})}{4R_{2}^{2}} - \frac{1}{4k^{2}}\cdot \sum\limits_{1 \le i,j \le 2k}\frac{\Nc(f_{\rho};\Dc_{R_{2};i,j})}{R_{1}^{2}}\right|\right] \\&+
\E\left[\left| \frac{1}{4k^{2}}\cdot \sum\limits_{1 \le i,j \le 2k}\frac{\Nc(f_{\rho};\Dc_{R_{2};i,j})}{R_{1}^{2}}
- c_{NS}(\rho)  \right|\right] < \xi+\epsilon + O\left(\frac{1}{R_{1}}+ \frac{R_{1}}{R_{2}} \right).
\end{split}
\end{equation}
by Lemma \ref{lem:Kac Rice flips},
\eqref{eq:E[|N(R2)/4R2^2-sum|]<<1/R1+R1/R2}, \eqref{eq:E[sum-cNS] sm triangle}, and, again, the triangle inequality.
Since $\epsilon>0$ on the r.h.s. of \eqref{eq:E[|N2/4R2^2-cNS|]<1/R1+R1/R2+xi+eps}
is arbitrary, fixing $R_{1}\gg 0$ satisfying \eqref{eq:N(R1)/4R^2-c<xi+eps} arbitrarily big, and taking
$\limsup\limits_{R_{2}\rightarrow\infty}$ of both sides of \eqref{eq:E[|N2/4R2^2-cNS|]<1/R1+R1/R2+xi+eps} yields
\begin{equation*}
\limsup\limits_{R\rightarrow\infty}\E\left[\left|\frac{\Nc(f_{\rho};R)}{4R^{2}} - c_{NS}(\rho)\right|\right] \le \xi;
\end{equation*}
comparing the latter equality with \eqref{eq:xi=liminf def} finally yields the existence of the limit \eqref{eq:dNS lim abs def}.

\end{proof}

\subsection{Proofs of the local estimates}
\label{sec:loc proofs}

\begin{proof}[Proof of Lemma \ref{lem:covar der basic}]

Let $r_{F}(x,y)= \E[F(x)\cdot F(y)]$ be the covariance function of $F$, by the assumptions of Lemma \ref{lem:covar der basic} we
have that
\begin{equation}
\label{eq:r(x,x)==1}
\E[f(x)\cdot f(x)]=r_{F}(x,x)\equiv 1.
\end{equation}
The independence of $F(x)$ and $\nabla F(x)$ then follows upon differentiating \eqref{eq:r(x,x)==1} concluding the first part of
Lemma \ref{lem:covar der basic}.
The second part of Lemma \ref{lem:covar der basic} is obvious from the definition \eqref{eq:Crho def} of $C(\rho)$ bearing in mind
the aforementioned diagonalisation of $C(\rho)$ by a rotation (see the interpretation of $C(\rho)$ and $\Pc_{\epsilon}$
immediately after \eqref{eq:Pesp def}).

\end{proof}

\begin{proof}[Proof of Lemma \ref{lem:Kac Rice intersections}]

Let $\gamma:[0,\Lc] \rightarrow\R^{2}$ be an arc-length parametrization of $\Cc$, and
$$g(t)=g_{\Cc,\rho}(t)=f_{\rho}(\gamma(t))$$ be the restriction $g:[0,\Lc]\rightarrow\R$ of $f_{\rho}$ along $\Cc$.
The process $g$ is centred Gaussian, with covariance function
\begin{equation}
\label{eq:rg restrict curve}
r_{g}(t_{1},t_{2}) = r_{\rho}(\gamma(t_{2})-\gamma(t_{1}))
\end{equation}
%with $r_{\rho}=r_{f_{\rho}}$ the covariance function of $f_{\rho}$;
%by
with $r_{\rho}$ the covariance function of $f_{\rho}$.
% ; by
% the definition $r_{\rho}$ is the (real valued)
% Fourier transform \eqref{eq:r_rho Fourier} of $\rho\in\Pc$.

The number of nodal intersections of $f_{\rho}$ against $\Cc$ is then
a.s. equal to $\Nc(f_{\rho},\Cc)=\Zc_{g}$,
the number of zeros of $g$ on $[0,\Lc]$. Since $f_{\rho}$ has unit
variance, so does $g$; therefore (Lemma \ref{lem:covar der basic}) for every $t\in [0,\Lc]$ the value
$g(t)$ is independent of the derivative $g'(t)$. We then have by
Kac-Rice ~\cite[Theorem 6.3]{AW}
\begin{equation*}
\E[\Z_{g}] = \int\limits_{0}^{\Lc} K_{1}(t)dt,
\end{equation*}
where
\begin{equation*}
K_{1}(t)=K_{1;g}(t) = \frac{1}{\pi}\sqrt{\frac{\partial^{2}r_{g}}{\partial t_{1}\partial t_{2}}\bigg|_{t_{1}=t_{2}=t}}
\end{equation*}
is the zero density of $g$. The statement of Lemma \ref{lem:Kac Rice intersections} will follow once we show that the mixed
second derivative $\frac{\partial^{2}r_{g}}{\partial t_{1}\partial t_{2}}$ of $r_{g}$ is uniformly bounded by an absolute constant,
independent of $\gamma$ and $\rho\in\Pc$.

To this end we differentiate \eqref{eq:rg restrict curve} to compute
\begin{equation*}
\partial_{t_{1}}\partial_{t_{2}} r_{g}(t_{1},t_{2}) = -\dot{\gamma}(t_{1})\cdot H_{r_{\rho}} (\gamma(t_{2})-\gamma(t_{1}))\cdot \dot{\gamma}(t_{2})^{t},
\end{equation*}
where $H_{r_{\rho}}$ is the Hessian of
$r_{\rho}$. That
$$\partial_{t_{1}}\partial_{t_{2}} r_{g}(t_{1},t_{2})$$ is bounded by
an absolute constant then follows from the fact
that $$\|\dot{\gamma}(t_{1})\| = \|\dot{\gamma}(t_{2})\|=1,$$ and that
$H_{r_{\rho}}$ is bounded follows by differentiating
\eqref{eq:r_rho Fourier}, using the bounded support of $\rho$.

\end{proof}

\begin{proof}[Proof of Lemma \ref{lem:Kac Rice flips}]
To prove the first assertion
we record the following useful fact about nondegenerate
centred Gaussians: with $(X,Y,Z)$ denoting components of a
nondegenerate multivariate normal distribution having mean zero, we
have
\begin{equation}
\label{eq:cond var shrink}
\var(X | Y=Z=0) \leq \var(X).
\end{equation}
While it is easy to validate \eqref{eq:cond var shrink} by an explicit computation, it is also
a (very) particular consequence of the vastly general Gaussian Correlation Inequality ~\cite{Roy}.

Now, by Kac-Rice ~\cite[Theorem 6.3]{AW} it follows that,
if for all $x\in\Dc$, the Gaussian distribution of
\begin{equation}
\label{eq:F flips def}
F(x):=(f_{\rho}(x),\partial_{1}f_{\rho}(x))
\end{equation}
is non-degenerate (holding by both parts of Lemma \ref{lem:covar der basic}),
then \eqref{eq:E[ZF]=intK(1) gen} is satisfied
with $$K_{1}(x)=K_{1;\rho}(x)$$ the appropriately defined flips
density \eqref{eq:K1 density def} with $F$ given by
\eqref{eq:F flips def}, and by stationarity we have
\begin{equation}
\label{eq:K1(x)=K1(0) stat}
K_{1}(x)\equiv K_{1}(0).
\end{equation}

Now from \eqref{eq:K1(x)=K1(0) stat} and \eqref{eq:E[ZF]=intK(1) gen}
it then follows that
\begin{equation}
\label{eq:E[flips] Kac-Rice stat}
\E[S_{1}(f_{\rho};\Dc)] = K_{1}(0)\cdot \area(\Dc),
\end{equation}
and it is sufficient to show that $$K_{1}(0) = O\left(\var(\partial_{2}
f_{\rho})^{1/2}\right).$$
Upon recalling that $F$ is
given by \eqref{eq:F flips def} we have that
\begin{equation}
\label{eq:K1 flips def}
K_{1}(0) = \phi_{F(0)}(0,0)\cdot \E[|\det J_{F}(0)|\big| f_{\rho}(0)=\partial_{1} f_{\rho}(0) = 0],
\end{equation}
where $\phi_{F(0)}$ is the probability density of the Gaussian vector
$$
F(0)=
(f_{\rho}(0),\partial_{1} f_{\rho}(0)),
$$
and
$$
J_{F}(x) =
\left(\begin{matrix} \partial_{1}f_{\rho}(x) &\partial_{2}f_{\rho}(x)
\\ \partial_{1}^{2}f_{\rho}(x) &\partial_{1}\partial_{2}f_{\rho}(x)
\end{matrix}\right)
$$
is the Jacobian matrix of $F$.

Conditioned on $\partial_{1} f_{\rho}(0) = 0$ we have that
\begin{equation*}
\det (J_{F}(x)) = -\partial_{2}f_{\rho}(x)\cdot \partial_{1}^{2}f_{\rho}(x),
\end{equation*}
hence \eqref{eq:K1 flips def} is
\begin{equation}
\label{eq:K1 flips C-S}
\begin{split}
K_{1}(0) &=
\frac{1}{2\pi \sqrt{\var(\partial_{1} f_{\rho}(0))}} \cdot
\E[|\partial_{2}f_{\rho}(0)\cdot \partial_{1}^{2}f_{\rho}(0)|
\big| f_{\rho}(0)=\partial_{1} f_{\rho}(0) = 0]
\\
&\le \frac{1}{2\pi \sqrt{\var(\partial_{1} f_{\rho}(0))}} \cdot
\sqrt{\var(\partial_{2}f_{\rho}(0)\big| f_{\rho}(0)=\partial_{1}
  f_{\rho}(0) = 0)}\times
\\&\times \sqrt{\var(\partial_{1}^{2} f_{\rho}(0)\big|
  f_{\rho}(0)=\partial_{1} f_{\rho}(0) = 0)}
\\&= O\left(
\frac{
\sqrt{\var(\partial_{2}f_{\rho}(0))} \times
\sqrt{ \var(\partial_{1}^{2} f_{\rho}(0))  }
}
{\sqrt{\var(\partial_{1} f_{\rho}(0))}}\right)
\end{split}
\end{equation}
by Cauchy-Schwartz and the above mentioned bound \eqref{eq:cond var shrink} on the
conditional variance.

Now, differentiating \eqref{eq:r_rho Fourier} we have that
\begin{equation}
\label{eq:var der int Fourier}
\var(\partial_{1}f_{\rho}(0))
= (2\pi)^{2}
\int_{B(1)} y_{1}^{2} d  \rho(y) \text{  and  }
\var(\partial_{2}f_{\rho}(0))
= (2\pi)^{2}
\int_{B(1)} y_2^{2} d  \rho(y),
\end{equation}
showing in particular the uniform bound
\begin{equation}
\label{eq:var(f2)=O(1)}
\var(\partial_{2}  f_{\rho}(0)) = O(1).
\end{equation}
%is uniformly bounded.
Differentiating \eqref{eq:r_rho Fourier} in a similar fashion we obtain
the analogous expression
\begin{equation}
\label{eq:f11 var Fourier}
\var(\partial_{1}^{2}f_{\rho}(0)) =
(2\pi)^{4}
\int_{B(1)} y_1^{4} d  \rho(y),
\end{equation}
for $\var(\partial_{1}^{2}f_{\rho}(0))$. The identity \eqref{eq:f11 var Fourier} together with \eqref{eq:var der int Fourier}
imply that the ratio
\begin{equation}
\label{eq:ratio 2nd der/1st der}
\frac{\var(\partial_{1}^{2}f_{\rho}(0))}
{\var(\partial_{1}  f_{\rho}(0))} = O(1)
\end{equation}
is uniformly bounded, since $y_1^4 \leq y_1^{2}$ for all $y \in B(1)$.
Finally \eqref{eq:ratio 2nd der/1st der} together with \eqref{eq:var(f2)=O(1)} imply that the r.h.s.
of \eqref{eq:K1 flips C-S} is uniformly bounded, sufficient for the first assertion of
Lemma \ref{lem:Kac Rice flips} via \eqref{eq:E[flips] Kac-Rice stat}.

The second assertion of Lemma \ref{lem:Kac Rice flips} can be deduced from the first by changing
coordinates via rotating by $\pi/2$.
The final assertion follows immediately from the two first.

\end{proof}

\section{Proof of Theorem \ref{thm:cNS cont}: continuity of the Nazarov-Sodin constant}

We shall treat the case of limiting spectral measures $\rho$ lying
in $\Pzero$ separately, and we begin with the following result.
\begin{lemma}
\label{lem:degenerate-limit}
If $\rho \in \Pzero$ and $\rho_i \Rightarrow \rho$ (convergence in weak-* topology), then
$$
c_{NS}(\rho_i) \to c_{NS}(\rho) = 0.
$$
\end{lemma}
\begin{proof}
By Lemma~\ref{lem:degenerate} we have $c_{NS}(\rho) = 0$.  Moreover,
the same holds for those $j$ such that $\rho_j \in \Pzero$ and
hence it is enough to treat the case that $\rho_j \not \in \Pzero$
for all $j$. Now, as $\rho_j \to \rho$ and $\rho \in \Pzero$, we
find that given $\epsilon>0$ there exist $J$ such that
$\rho_{j} \in \Peps$ for all $j \geq J$.  Thus, after making a
(possibly $j$-dependent) rotational change of variables, we may
assume that $\var(\partial_{1} f_{\rho_{j}}) \leq \epsilon$ and
Lemma~\ref{lem:Kac Rice flips} then implies that
$c_{NS}(\rho_j) = O\left(\sqrt{\epsilon}\right)$ for $j \geq J$. The result follows.

\end{proof}

\subsection{Preliminary results}

\subsubsection{Perturbing the random field}
\label{sec:stability}

The following proposition, proved in \S~\ref{sec:prop stab bnd proof}
below, will be used in the proof of Theorem \ref{thm:cNS cont}.
%below, will be used while proving Theorem \ref{thm:cNS cont}.

\begin{proposition}
\label{prop:stability bound}
Let $R>0$ be sufficiently big, $\epsilon>0$, $\xi>0$, and let
$\{\rho_{j}\}\subseteq \Pc\setminus\Peps$ be a sequence of probability
measures,
weak-$*$ convergent to $\rho_{0}\in\Pc\setminus\Peps$.
There exists a number $j_{0}=j_{0}(\epsilon;\{\rho_{j}\};R,\xi)>0$ such that for all $j>j_{0}$
there exists a coupling of $f_{\rho_{j}}$ and $f_{\rho_{0}}$ and an event
$\Omega_{0}=\Omega_{0}(\rho_{0},\rho_{j};R,\xi)$ of probability $\prob(\Omega_{0})<\xi$ such that
on $\Omega\setminus\Omega_{0}$ we have
\begin{equation}
\label{eq:Ns sandwich perturb}
\Nc(f_{\rho_{j}};R-1) \le \Nc(f_{\rho_{0}};R) \le \Nc(f_{\rho_{j}};R+1).
\end{equation}
\end{proposition}

\subsubsection{Small domains}
\label{sec:small dom}

For smooth (deterministic) function $F:\R^{2}\rightarrow \R$, $R>0$
and a small parameter $\delta \in
(0,1]$ we denote $\Nc_{\delta-sm}(F;R)$ to
be the number of domains of area $<\delta$ (``$\delta$-small
domains'') lying entirely inside $B(R)$. Accordingly, let
$$\Nc_{\delta-big}(F;R)=\Nc(F;R)-\Nc_{\delta-sm}(F;R)$$
be the number of ``$\delta$-big domains'' (a more appropriate, though
cumbersome, term would be ``$\delta$-not-small'').
We have the following bound for the expected number of $\delta$-small
domains of $f_{\rho}$.

% \begin{lemma}[Cf. \cite{NS 2009}; \cite{SW} Lemma $7.5$]
\begin{lemma}[Cf. \cite{NS 2009}; \cite{SW} Lemma $4.12$]
\label{lem:small domains bnd}
For every $\epsilon>0$ there exist constants $c_{0}(\epsilon),C_{0}(\epsilon)>0$
such that the expected number of $\delta$-small domains satisfies
 \begin{equation*}
%\sup\limits_{R\rightarrow\infty}
\sup\limits_{R \geq 10}
\frac{\E[\Nc_{\delta-sm}(f_{\rho};R)]}{R^{2}}
\le
C_{0}(\epsilon)\cdot \delta^{c_{0}(\epsilon)}
\end{equation*}
uniformly for all
$\rho \in\Pc \setminus \Peps$.
\end{lemma}

\begin{proof}

If $\rho\in\Pc\setminus\Pc_{0}$ then the non-degeneracy conditions
of \cite[Lemma 9]{NS 2009} (or ~\cite[Proposition 4.12]{SW}) are satisfied, hence these imply
that in this case there exist constants $C_{0}(\rho)$ and $c_{0}(\rho)$, depending continuously on $\rho\in \Pc\setminus\Pc_{0}$, so that
\begin{equation*}
%\limsup\limits_{R\rightarrow\infty}
\sup\limits_{R \geq 10}
\frac{\E[\Nc_{\delta-sm}(f_{\rho};R)]}{R^{2}}
\le C_{0}(\rho)\cdot \delta^{c_{0}(\rho)}.
\end{equation*}
The uniformity for choice of $C(\epsilon),c(\epsilon)$ then follows
from the compactness of $\Pc\setminus\Pc_{\epsilon}\subseteq
\Pc\setminus\Pc_{0}.$

\end{proof}

\subsection{Proof of Theorem \ref{thm:cNS cont}}

\begin{proof}[Proof of Theorem \ref{thm:cNS cont}]

Let $\{\rho_{j}\}_{j\ge 1}\subseteq \Pc$ be a sequence of probability measures weak-$*$ converging to $\rho_{0}\in\Pc$;
the statement of Theorem \ref{thm:cNS cont} is that in this situation the corresponding Nazarov-Sodin constants
\begin{equation}
\label{eq:cNS(rhoj)->cNS(rho0)}
c_{NS}(\rho_{j})\rightarrow c_{NS}(\rho_{0})
\end{equation}
converge to the Nazarov-Sodin constant of $\rho_{0}$.

The case of $\rho_0 \in \Pzero$ follows from
Lemma \ref{lem:degenerate-limit}. For $\rho_0 \not \in \Pzero$, we
have $\rho_{0} \in \Pc \setminus \Pc_{2\epsilon}$ for some $\epsilon>0$ and
thus we may assume that $\rho_j \in \Pc \setminus \Pc_{\epsilon}$ for all
sufficiently large $j$; without loss of generality we may assume that
$\rho_j \in \Pc \setminus \Pc_{\epsilon}$ for all $j$.

Proposition \ref{prop:NS constant uniform} yields that given
$\alpha>0$ there exists $$R_{0}=R_{0}(\alpha)\gg 0$$ sufficiently
big so that for all $R>R_{0}$ and all $\rho\in \Pc$ we
have
%\fixme{Verify that there are no exclusions, such as supported on
%  $2$ antipodal points.}
\begin{equation}
\label{eq:E[fa,R]/R^2-cNS(a)<epsilon}
\left|\frac{\E[\Nc(f_{\rho},R)]}{4R^{2}} - c_{NS}(\rho)\right|  < \alpha;
\end{equation}
in particular \eqref{eq:E[fa,R]/R^2-cNS(a)<epsilon} applies to $\rho=\rho_{j}$ with $j\ge 1$ or $\rho=\rho_{0}$.
We now apply Proposition \ref{prop:stability bound} with $R>R_{0}$, and $\xi>0$ small, so that it yields
a number $j_{0}$ sufficiently big such that for all $j>j_{0}$ there exists an event $\Omega_{0}=\Omega_{0}(R,j,\xi)$
of probability
\begin{equation}
\label{eq:prob(Om/Om0)<xi}
\prob(\Omega_{0})>1-\xi,
\end{equation}
such that on $\Omega_{0}$ we have
\begin{equation}
\label{eq:N(fa0,R-1)<=N(fa,R)<=N(fa0,R+1)}
\Nc(f_{\rho_{j}};R-1) \le \Nc(f_{\rho_{0}};R) \le \Nc(f_{\rho_{j}};R+1).
\end{equation}

We are now going to show that the
difference $$\E[\Nc(f_{\rho_{j}};R+1)] -
\E[\Nc(f_{\rho_{j}};R-1)]\ge 0$$ is  small (compared to $R^{2}$)
for $R\rightarrow\infty$;
this would also imply that $$\E[\Nc(f_{\rho_{j}};R)] -
\E[\Nc(f_{\rho_{0}};R)]$$ is small (compared to $R^{2}$), and thus
$c_{NS}(\rho_{j})-c_{NS}(\rho_{0})$
is small via \eqref{eq:E[fa,R]/R^2-cNS(a)<epsilon}.
Recall that $\Dc_{R}$ is the square $$\Dc_{R}:=[-R,R]^{2}\subseteq
\R^{2},$$ and denote
\begin{equation*}
\Ac_{R}=\Dc_{R+1}\setminus \operatorname{Int}(\Dc_{R-1})
\end{equation*}
to be the strip lying inside the $2(R+1)$-side square, outside the $2(R-1)$-side square.
If for some $\rho\in\Pc$ a nodal domain of $f_{\rho}$ is lying
entirely inside $\Dc_{R+1}$ but not $\Dc_{R-1}$, then that nodal
domain is either entirely lying
inside $\Ac_{R}$ or intersects the boundary $\partial \Dc_{R-1}$ of
the smaller of the squares. In either case that nodal domain
necessarily contains either a horizontal or a vertical flip lying in
$\Ac_{R}$, i.e. a point $x\in \Ac_{R}$ such that either
$f_{\rho}(x)=f_{\rho;1}(x)=0$ or $f_{\rho}(x)=f_{\rho;2}(x)=0$, that
is, recalling the notation \eqref{eq:S1,S2 def} of nodal
flips numbers, we have
\begin{equation*}
 0 \leq 
\Nc(f_{\rho};R+1) - \Nc(f_{\rho};R-1) \le S_{1}(f_{\rho};\Ac_{R})+S_{2}(f_{\rho};\Ac_{R}),
\end{equation*}
and upon taking the expectations of both sides of the latter inequality we obtain
\begin{equation}
\label{eq:EN(R+1)-N(R-1)rho=O(R)}
\begin{split}
 0 \leq
\E[\Nc(f_{\rho};R+1)] - \E[\Nc(f_{\rho};R-1)] \le
\E[S_{1}(f_{\rho};\Ac_{R})+S_{2}(f_{\rho};\Ac_{R})]
% = O(R),
 \leq C_{1} \cdot R
\end{split}
\end{equation}
by Lemma \ref{lem:Kac Rice flips},
with $C_{1}>0$ an absolute constant.
%with constant involved in the $`O'$-notation absolute.

Now let $\delta>0$ be a small parameter
and recall the definition of $\delta$-small and $\delta$-big domains counts in \S~\ref{sec:small dom}.
We invoke \eqref{eq:N(fa0,R-1)<=N(fa,R)<=N(fa0,R+1)} via
\eqref{eq:EN(R+1)-N(R-1)rho=O(R)}, together with Lemma~\ref{lem:small
  domains bnd}, and obtain that (for $j > j_{0}$)
\begin{equation*}
\begin{split}
&\E\left[\left|\Nc(f_{\rho_{j}};R) -\Nc(f_{\rho_{0}};R)\right|\right]
\\
% &= O(R)+ O(R^{2}\delta^{c_{0}}) +
% \int\limits_{\Omega\setminus\Omega_{0}} \left(
%   \Nc_{\delta-big}(f_{\rho_{0}};2R) +
%   \Nc_{\delta-big}(f_{\rho_{j}};2R)\right)d\prob(\omega)
&\leq
 C_{2}  \left(
    R+ R^{2}\delta^{c_{0}} +
 \int\limits_{\Omega\setminus\Omega_{0}} \left(
   \Nc_{\delta-big}(f_{\rho_{0}};2R) +
   \Nc_{\delta-big}(f_{\rho_{j}};2R)\right)d\prob(\omega)
\right)
\\
&\leq C_{2} \left(
R+ R^{2}\delta^{c_{0}} + \frac{8}{\delta}R^{2}\cdot
  \prob(\Omega\setminus\Omega_{0})\right).
%
% &= O(R)+ O(R^{2}\delta^{c_{0}}) + O\left(\frac{1}{\delta}R^{2}\cdot
%   \prob(\Omega\setminus\Omega_{0})\right).
\end{split}
\end{equation*}
for $C_{2}>0$ an absolute constant.
Recalling \eqref{eq:prob(Om/Om0)<xi} the above implies
\begin{equation*}
\frac{\E\left[\left|\Nc(f_{\rho_{j}};R)
      -\Nc(f_{\rho_{0}};R)\right|\right]}{R^{2}}
\leq C_{2} \left(\frac{1}{R}+\delta^{c_{0}}+\frac{8\xi}{\delta}\right).
%
%\frac{\E\left[\left|\Nc(f_{\rho_{j}};R) -\Nc(f_{\rho_{0}};R)\right|\right]}{R^{2}}  =O\left(\frac{1}{R}+\delta^{c_{0}}+\frac{\xi}{\delta}\right).
\end{equation*}
Now using the triangle inequality with
\eqref{eq:E[fa,R]/R^2-cNS(a)<epsilon} applied on $\rho_{j}$ and
$\rho_{0}$ implies that for $j>j_{0}$ one has
\begin{equation}
\label{eq:cNs(j)-cNS(j0) bound}
|c_{NS}(\rho_{j})-c_{NS}(\rho_{0})|  \leq C_{2}
  \left(\frac{1}{R}+\delta^{c_{0}}+\frac{8\xi}{\delta}+2\alpha\right).
% |c_{NS}(\rho_{j})-c_{NS}(\rho_{0})| = O\left(\frac{1}{R}+\delta^{c_{0}}+\frac{\xi}{\delta}+\alpha\right).
\end{equation}
Since the r.h.s. (and thus the l.h.s.) of \eqref{eq:cNs(j)-cNS(j0)
  bound} can be made arbitrarily small by first
choosing
the parameters $\alpha$ and $\delta$ sufficiently small, and then $R>R_{0}(\alpha)$ sufficiently
large, and finally $\xi$ sufficiently small,
and in light of the fact that the l.h.s. of \eqref{eq:cNs(j)-cNS(j0) bound} does not depend on $R$,
this yields \eqref{eq:cNS(rhoj)->cNS(rho0)}.
As mentioned above, this is equivalent to the statement of Theorem
\ref{thm:cNS cont}.

\end{proof}

\begin{remark}
\label{rem:small-faber-krahn}
The above argument can be simplified in the case of monochromatic
waves, as here small domains do not exist by an application of
the Faber-Krahn inequality
~\cite[Theorem 1.5]{Ma}, so there is no need to invoke Lemma
\ref{lem:small domains bnd} to bound their contribution.
\end{remark}

\section{Proof of Proposition \ref{prop:stability bound}}
\label{sec:prop stab bnd proof}

The ultimate goal of this section is giving a proof for Proposition \ref{prop:stability bound}. Towards this goal
we first construct the exceptional event $\Omega_{0}$ in \eqref{eq:Omega0 excp def} below; it will consist of
various sub-events defined in \S~\ref{sec:Omega0 constr} that would guarantee that on $\Omega\setminus\Omega_{0}$
both fields $f_{\rho_{0}}$ and $f_{\rho_{j}}$ (for $j$ sufficiently big)
are ``stable" in the sense that a small perturbation of our function has a minor effect on its nodal structure,
and that the perturbation $f_{\rho_{0}}-f_{\rho_{j}}$ is
small in a sense to be
made precise. That $\Omega_{0}$ is {\em rare}
is established in \S~\ref{sec:Omega0 rare}. Proposition \ref{prop:stability bound} will be finally proved in \S~\ref{sec:stability bound prf}
assuming some auxiliary results that will be established in \S~\ref{sec:aux lemm stab}.

\subsection{Constructing the exceptional event $\Omega_{0}$}

\label{sec:Omega0 constr}

\begin{definition}

\begin{enumerate}

\item

For $R>0$ big parameter, $\beta>0$ small parameter, and $\rho\in\Pc$ we define the ``unstable" event
\begin{equation*}
\Omega_{1}(f_{\rho};R,\beta) := \left\{ \min\limits_{x\in B(2R)}\max \{ |f_{\rho}(x)|, \| \nabla f_{\rho} (x)\| \}  \le 2\beta\right\},
\end{equation*}
i.e., that there exists a point in the ball $B(2R)$ such that both
$f_{\rho}$ and its gradient are small.

\item

For $R,M>0$ big parameters, $\rho\in \Pc$ we define
\begin{equation*}
\Omega_{2}(f_{\rho};R,M) := \left\{ \|f_{\rho} \|_{C^{2}(B(2R))}  \ge M \right\}.
\end{equation*}

\item
Let $\rho,\rho'\in\Pc$ be two measures and $f_{\rho}, f_{\rho'}$
copies of the corresponding random fields on $\R^{2}$
defined on the same probability space $\Omega$. For $R>0$, $\beta>0$ define
\begin{equation*}
\Omega_{3}(f_{\rho},f_{\rho'};R,\beta) := \left\{ \|f_{\rho}-f_{\rho'} \|_{C^{1}(B(2R))}  \ge \beta \right\}.
\end{equation*}

\end{enumerate}

\end{definition}

\subsection{The exceptional event is rare}
\label{sec:Omega0 rare}

We present the following auxiliary lemmas \ref{lem:Omega1
  small}-\ref{lem:Omega3 small} which together imply that the exceptional
event is rare. Lemmas \ref{lem:Omega1 small}-\ref{lem:Omega3 small}
will be proved in \S~\ref{sec:aux lemm stab}.

\begin{lemma}[Cf. \cite{So}, Lemma $5$]
\label{lem:Omega1 small}
For every $\rho\in\Pc\setminus\Peps$, $R>0$, $M>0$ and $\xi>0$ there exists a number $\beta=\beta(\epsilon;\rho;R,\xi)>0$ such that the probability
of $\Omega_{1}(f_{\rho};R,\beta)$ outside of $\Omega_{2}(f_{\rho};R,M)$ is
\begin{equation*}
\prob(\Omega_{1}(f_{\rho};R,\beta)\setminus \Omega_{2}(f_{\rho};R,M)) < \xi.
\end{equation*}
\end{lemma}

\begin{lemma}
\label{lem:Omega2 small}

\begin{enumerate}

\item
\label{it:Omega2(rho0) small}
For every $\rho\in\Pc$, $R>0$ and $\xi>0$ there exists a number $M=M(f_{\rho};R,\xi)$ so that
\begin{equation*}
\prob(\Omega_{2}(f_{\rho};R,M))< \xi.
\end{equation*}

\item
\label{it:Omega2(rhoj) small}
Let $R>0$ be sufficiently big, $\xi>0$, and a sequence $\{\rho_{j}\}\subseteq \Pc$ of probability measures, weak-$*$ convergent to
$\rho_{0}\in\Pc$. Then there exists a number $M=M(\rho_{0};R,\xi)>0$ and $j_{0}=j_{0}(f_{\rho_{0}};R,\xi)$ such that for all $j>j_{0}$ we have
\begin{equation*}
\prob(\Omega_{2}(f_{\rho_{j}};R,M))< \xi.
\end{equation*}

\end{enumerate}

\end{lemma}

\begin{lemma}[Cf. \cite{So}, Lemma $4$]
\label{lem:Omega3 small}
Let $R>0$ be sufficiently big, $M>0$, $\beta>0$, $\xi>0$, and a sequence $\{\rho_{j}\}\subseteq \Pc$ of probability measures, weak-$*$ convergent to
$\rho_{0}\in\Pc$. There exists a number $j_{0}=j_{0}(\{\rho_{j}\};R,\xi)>0$ such that for all $j>j_{0}$ there exists
a coupling of $f_{\rho_{j}}$ and $f_{\rho_{0}}$ such that the probability of $\Omega_{3}(f_{\rho_{0}},f_{\rho_{j}};R,\beta)$
outside $$\Omega_{2}(f_{\rho_{0}};R,M)\cup \Omega_{2}(f_{\rho_{j}};R,M)$$ is
\begin{equation*}
\prob(\Omega_{3}(f_{\rho_{0}},f_{\rho_{j}};R,\beta)\setminus (\Omega_{2}(f_{\rho_{0}};R,M)\cup \Omega_{2}(f_{\rho_{j}};R,M))) < \xi.
\end{equation*}
\end{lemma}

\subsection{Proof of Proposition \ref{prop:stability bound}}
\label{sec:stability bound prf}

For consistency with the earlier works the various events $\Omega_{i}$ in \S~\ref{sec:Omega0 constr}
are defined in terms of properties of the relevant random fields imposed on balls of large radius;
this is slightly inconsistent to the nodal counts $\Nc(\cdot;\cdot)$ in our main results that are defined on large
squares. This however will not require any extra work due to the fortunate fact that the squares are contained in
slightly bigger balls.

The following lemma states that, under the ``stability assumption'' on
a function,
its nodal components are stable.

\begin{lemma}[~\cite{So}, lemmas 6-7]
\label{lem:stability comp}
Let $\beta>0$ be a small number, $$\Dc=\Dc_{R+1}\subseteq \R^{2}$$ the side-$2(R+1)$ square,
and $f\in C^{1}(B)$ be a smooth function on $\Dc$ such that
$$\min\limits_{x\in \Dc} \max\{ f(x), \|\nabla f(x)\|\} > \beta.$$ Suppose that $g\in C(\Dc)$
is a continuous function on $\Dc$ such that $$\sup\limits_{x\in B} |g(x)|<\beta.$$ Then every
nodal component $\gamma$ of $f$ lying entirely in $\Dc_{R}$ generates a unique nodal component $\widetilde{\gamma}$ of $(f+g)$
lying in $\Dc_{R+1}$ with distance $d(\gamma,\widetilde{\gamma}) < 1$ from $\widetilde{\gamma}$
(in fact,  the stronger statement $$\widetilde{\gamma}\subseteq
\gamma_{1}=\{x\in \Dc_{R+1}:\: d(x,\gamma)<1\}$$
holds); different components of $f$ correspond to different components of $(f+g)$.
\end{lemma}

\begin{proof}[Proof of Proposition \ref{prop:stability bound}]

Let $R>0$, $\xi>0$, $\{\rho_{j}\}\subseteq \Pc$, and $\rho_{0}\in\Pc$ be given.
An application of Lemma \ref{lem:Omega2 small}, part \ref{it:Omega2(rho0) small} on $(\rho_{0},R,\xi/4)$ and
part \ref{it:Omega2(rhoj) small} on $(\{\rho_{j}\},\rho_{0};R,\xi/8)$ yield a number $M>0$ (a priori two
different numbers that could be replaced by their maximum) such that both
\begin{equation}
\label{eq:prob(Omega2(rho0,rhoj))<xi/4}
\prob(\Omega_{2}(\rho_{0};R,M) < \frac{\xi}{8} \text{ and } \prob(\Omega_{2}(\rho_{j};R,M) < \frac{\xi}{8},
\end{equation}
for $j>j_{0}$ sufficiently big.
An application of
Lemma \ref{lem:Omega1 small} on $(\rho_{0},R,M,\xi/4)$ yields a number $\beta>0$
sufficiently small so that
\begin{equation}
\label{eq:prob(Omega1)<xi/4}
\prob(\Omega_{1}(f_{\rho_{0}};R,\beta) < \frac{\xi}{4} + \prob(\Omega_{2}(f_{\rho_{0}};R,M)) < \frac{\xi}{2},
\end{equation}
by \eqref{eq:prob(Omega2(rho0,rhoj))<xi/4}.
Finally, an application of Lemma \ref{lem:Omega3 small} on $(\{\rho_{j}\},\rho_{0};R,M,\beta,\xi/4)$
yields a coupling of $(f_{\rho_{0}},f_{\rho_{j}})$ such that for all $j>j_{0}$ we have
\begin{equation}
\label{eq:prob(Omega3)<xi/4}
\prob(\Omega_{3}(f_{\rho_{0}},f_{\rho_{j}};R,\beta)) < \frac{\xi}{4}+ \prob(\Omega_{2}(f_{\rho_{0}};R,M))+
\prob(\Omega_{2}(f_{\rho_{j}};R,M)) < \frac{\xi}{2},
\end{equation}
again by \eqref{eq:prob(Omega2(rho0,rhoj))<xi/4}.

Let
\begin{equation}
\label{eq:Omega0 excp def}
\Omega_{0}:= \Omega_{1}(f_{\rho_{0}};R,\beta)\cup\Omega_{3}(f_{\rho_{0}},f_{\rho_{j}};R,\beta)
\end{equation}
of probability
\begin{equation}
\label{eq:prob(Omega0)<xi}
\prob(\Omega_{0}) < \xi
\end{equation}
by \eqref{eq:prob(Omega1)<xi/4} and \eqref{eq:prob(Omega3)<xi/4}, provided that $j$ is sufficiently big.
On $\Omega\setminus\Omega_{0}$ the function $f_{\rho_{0}}$ is stable in the sense that
\begin{equation}
\label{eq:frho0 stable}
\min\limits_{x\in B(2R)}\max \{ |f_{\rho_{0}}(x)|, \| \nabla f_{\rho_{0}} (x)\| \}  \ge 2\beta,
\end{equation}
and
\begin{equation}
\label{eq:|frho0-frhoj|<beta}
\left\{ \|f_{\rho_{0}}-f_{\rho_{j}} \|_{C^{1}(B(2R))}  \le \beta \right\}.
\end{equation}
Together \eqref{eq:frho0 stable} and \eqref{eq:|frho0-frhoj|<beta}
imply the stability of $f_{\rho_{j}}$, i.e., that
\begin{equation}
\label{eq:frhoj stable}
\min\limits_{x\in B(2R)}\max \{ |f_{\rho_{j}}(x)|, \| \nabla f_{\rho_{j}} (x)\| \}  \ge \beta
\end{equation}
via the triangle inequality. Note that for $R\gg 0$ sufficiently big $\Dc_{R+1}\subseteq B(2R)$ so that
all the above inequalities are satisfied in $\Dc_{R+1}$.

Now an application of Lemma \ref{lem:stability comp} with $f=f_{\rho_{0}}$ and $g=f_{\rho_{j}}-f_{\rho_{0}}$,
and upon bearing in mind \eqref{eq:frho0 stable} and \eqref{eq:|frho0-frhoj|<beta}
yields on $\Omega\setminus\Omega_{0}$ the r.h.s. of the inequality
\eqref{eq:Ns sandwich perturb}. The same argument
now taking $f=f_{\rho_{j}}$ and $g=f_{\rho_{0}}-f_{\rho_{j}}$, this time employing \eqref{eq:frhoj stable} and
\eqref{eq:|frho0-frhoj|<beta} yields on $\Omega\setminus\Omega_{0}$ the l.h.s. of \eqref{eq:Ns sandwich perturb}.
The above shows that \eqref{eq:Ns sandwich perturb} holds on $\Omega\setminus\Omega_{0}$, and
in addition \eqref{eq:prob(Omega0)<xi} provided that $j$ is
sufficiently big.
The proof of Proposition \ref{prop:stability bound} is concluded.

\end{proof}

\subsection{Proofs of the auxiliary lemmas \ref{lem:Omega1 small}-\ref{lem:Omega3 small}}
\label{sec:aux lemm stab}

We begin with the following simple lemma.
\begin{lemma}
  \label{lem:unif conv Fourier}
  Let $ \rho_1, \rho_2, \ldots \in \Pc$ be a sequence of
  spectral measures such that $\rho_{k}\Rightarrow \rho$, with the
  limiting measure $\rho \in \Pc$.  Then
$$\widehat{\rho_{k}}(\xi)\rightarrow \widehat{\rho}(\xi)$$ locally
uniformly, i.e. $r_{\rho_{k}}(x)\rightarrow r_{\rho}(x)$, uniformly on
compact subsets of $\R^{2}$. Moreover, the same holds for any (fixed)
finite number of derivatives.
\end{lemma}

\begin{proof}
  Let $D$ be the closure of the support of the spectral measure; we
  recall the assumption that $D$ is compact (this certainly holds for
  band-limited random waves, as well as for monochromatic waves.)
  Further, let $K \subset \R^2$ be compact. We note that the functions
  $\xi \to e( \xi \cdot x)$, as $x$ ranges over elements in
  $K$ is a uniformly continuous family. Moreover, as $D$ is compact and we
 consider probability measures on $D$, we find that
$$
\xi \to \widehat{\rho}(\xi) = \int_{D} e( \xi \cdot x ) d \rho(x)
$$
is uniformly continuous for all probability measures $\rho$ on $D$, and that
the Lipschitz estimate
$$
|\widehat{\rho}(\xi) - \widehat{\rho}(\xi')|
=O_{D} \left(|\xi - \xi'|\right)
$$
holds for all $\rho$.

Let $\alpha>0$ be given.
Given $\xi \in K$, choose $k(\xi)$ such that
$|\widehat{\rho_k}(\xi) - \widehat{\rho}(\xi)| < \alpha$ holds for all $k \geq
k(\xi)$.
Further, for each
$\xi \in K$ there exists an open ball $B_{\xi}$ centred at $\xi$ such
that
$$
|\widehat{\rho_k}(\xi') - \widehat{\rho_k}(\xi)| < \alpha
$$
for all $\xi' \in B_{\xi}$ and all $k$, and the same estimate holds
for $\widehat{\rho}$.

As $\{ B_{\xi} \}_{\xi \in K}$ is an open cover of the compact set
$K$, we find that
$K \subset \cup_{i=1}^{I} B_{\xi_{i}}$ for some finite collection of
points $\xi_1, \ldots, \xi_I$.  Define
$k = \max \{ k(\xi_{i}), i = 1, \ldots, I\}$.
If $\xi \in K$ there exist $i$ such that $\xi \in B_{\xi_{i}}$, and
thus, for $l \geq k$,
$$
|\widehat{\rho_l}(\xi) - \widehat{\rho}(\xi)|
=
|\widehat{\rho_l}(\xi) - \widehat{\rho_{l}}(\xi_{i})
+ \widehat{\rho_{l}}(\xi_{i}) - \widehat{\rho}(\xi_{i})
+ \widehat{\rho}(\xi_{i})  - \widehat{\rho}(\xi)
|
\leq 3 \alpha
$$
and hence the convergence is uniform in $\xi$.
Finally, a similar argument gives that the same holds for a finite number of
derivatives of $r_{\rho_{k}}$.

\end{proof}

\begin{proof}[Proof of Lemma \ref{lem:Omega1 small}]

The proof is very similar to the proof of ~\cite[Lemma $5$]{So}
presented in ~\cite[p. $23$]{So}. The independence of
$(f_{\rho}(x),\nabla f_{\rho}(x))$ in Lemma \ref{lem:covar der basic}
as well as the determinant of $C(\rho)$ in \eqref{eq:Crho def} being
bounded away from $0$ play a crucial role at the end of
the proof presented in ~\cite[p. 23]{So} in showing that both $f(x)$
and $\nabla f(x)$ being small is very rare.
\end{proof}

\begin{proof}[Proof of Lemma \ref{lem:Omega2 small}]

The proof is very similar to the proof of ~\cite[Lemma $6.6$]{SW}. Here to use the Sudakov-Fernique Comparison Inequality
we invoke Lemma \ref{lem:unif conv Fourier} so that the supremum of $f_{\rho_{j}}$ and its derivatives over a compact
domain is controlled by the supremum of $f_{\rho_{0}}$ and its respective derivatives over the same domain.

\end{proof}

\begin{proof}[Proof of Lemma \ref{lem:Omega3 small}]

We employ \cite[Lemma~$4$]{So} which states that the
conclusion of Lemma \ref{lem:Omega3 small} holds if
$r_{\rho_{j}}\rightarrow r_{\rho_{0}}$ locally uniformly together
with their finitely many derivatives, i.e. that for all multi-index
$J$,  and $|J|$ bounded,
\begin{equation*}
\sup\limits_{\|x\|\le 2R}| \partial_{J} r_{\rho_{j}}(x) - \partial_{J}r_{\rho_{0}}(x) | \rightarrow 0.
\end{equation*}
That this is so in our case follows from Lemma \ref{lem:unif conv Fourier}.

\end{proof}

%\end{lemma}

%\section{Proof of Proposition \ref{prop:cNS(Cil)=0}}
\section{Proof of Theorem~\ref{thm:NS const torus}: nodal count for arithmetic random waves}

\subsection{Proof of Theorem \ref{thm:NS const torus}}

We begin by the following lemma asserting that the Nazarov-Sodin constant vanishes for the (tilted) Cilleruelo measure.

\begin{lemma}
\label{lem:cNS(Cil)=0}
The Nazarov-Sodin constant of the Cilleruelo measure \eqref{eq:nu0 Cil
  meas} vanishes, i.e.,  $$c_{NS}(\nu_{0})=0.$$
\end{lemma}

Before proving Lemma \ref{lem:cNS(Cil)=0} in \S~\ref{sec:NS(Cil)=0 proof} we present the proof of Theorem \ref{thm:NS const torus}.

\begin{proof}[Proof of Theorem \ref{thm:NS const torus} assuming Lemma \ref{lem:cNS(Cil)=0}]

\noindent{\bf Proof of part \ref{it:N(fn)=c(mun)n+O(sqrt(n))}:}
We use the natural quotient map $q:\R^{2}\hookrightarrow \Tb^{2}$ and define the scaled random fields
$g_{n}:\R^{2}\rightarrow\R$ as $$g_{n}(y):= f_{n}(q(y/\sqrt{n})).$$ Then $g_{n}$ is a centred Gaussian random field
with spectral measure $\mu_{n}$ on $\Sc^{1}$, as in \eqref{eq:mun
  spect meas def} (one could also write $g_{n}=f_{\mu_{n}}$ though
we will refrain from doing it to avoid confusion). It is then clear that the nodal domains of $g_{n}$
lying inside the square $\Dc_{\sqrt{n}}= [-\sqrt{n}/2,\sqrt{n}/2]^{2}$
are in a $1-1$ correspondence with the nodal domains
of $f_{n}$ that do not intersect the image $q(\partial ([-1/2,1/2]^{2}))\subseteq \Tb^{2}$ of the boundary of the fundamental domain of $\Tb^{2}$.
Hence, under the notation of Lemma \ref{lem:Kac Rice intersections}, we have
\begin{equation}
\label{eq:N(fn) torus <-> square}
\left|\Nc(f_{n})-\Nc(g_{n};\sqrt{n}/2)\right| \le \Nc(g_{n},\Cc_{n}),
\end{equation}
where $\Cc_{n}$ is the boundary curve $\Cc_{n}=\partial [-\sqrt{n}/2,\sqrt{n}/2]$ of the side-$\sqrt{n}$ square. An application of
Lemma \ref{lem:Kac Rice intersections} then yields $$\Nc(f_{\mu_{n}},\Cc_{n}) = O(\sqrt{n}).$$
This, together with \eqref{eq:N(R)=cR^2+O(R)} and \eqref{eq:N(fn) torus <-> square} finally yields \eqref{eq:E[N]=cNS*n+O(sqrt(n))}.

\vspace{3mm}

\noindent{\bf Proof of part \ref{it:|N(fn)-c(mun)n|=O(sqrt(n))}}:

Since $\mu\in \Pc_{symm}$ and $\mu$ is assumed to have no atoms, then
$\mu$ satisfies the axioms $(\rho 1)-(\rho 3)$
of ~\cite{So}. Lemma \ref{lem:unif conv Fourier} then implies that, in
the language of ~\cite[Definition 1]{So}, the family $\{f_{n_{j}}\}$
of toral random fields has {\em translation invariant local limits}
$f_{\rho}$. Hence ~\cite[Theorem 4]{So} implies \eqref{eq:conv mean
  torus} (see also ~\cite[Theorem 1.2]{Ro}).

\vspace{3mm}

\noindent{\bf Proof of part \ref{it:cNS(rho)=0 <=> Cil}}:

First, if $\mu$ is neither the Cilleruelo measure $\nu_{0}$ in
\eqref{eq:nu0 Cil meas} nor the tilted Cilleruelo measure
$\widetilde{\nu_{0}}$ in \eqref{eq:tildnu0 tilt Cil}, then $\mu$ is
supported on at least four distinct pairs of antipodal points. Thus,
by  \cite[Remark~3]{In} (or ~\cite{SoPr}), $c_{NS}(\mu) >0$. Conversely,
the Nazarov-Sodin constant vanishes
$c_{NS}(\nu_{0})=c_{NS}(\widetilde{\nu_{0}})=0 $ for both the
Cilleruelo and tilted Cilleruelo measures by Lemma \ref{lem:cNS(Cil)=0} (which
is valid for $\widetilde{\nu_{0}}$ by rotation of $\pi/4$ symmetry).

\vspace{3mm}

\noindent{\bf Proof of part \ref{it:cNS(Psymm)=[0,dmax]}}:

Let $\Bc\subseteq\Sc_{symm}$ be the set of weak-$*$ partial limits of $\{\mu_{n}\}$; we claim that $\Bc$ is {\em connected}. Once having
the connectedness in our hands, part \ref{it:cNS(Psymm)=[0,dmax]} of Theorem \ref{thm:NS const torus} follows from
the continuity of $c_{NS}$ (Theorem \ref{thm:cNS cont}), vanishing $c_{NS}(\nu_{0})=0$ of the Nazarov-Sodin constant of the Cilleruello measure
(Theorem \ref{thm:NS const torus}, part \ref{it:cNS(rho)=0 <=> Cil}),
and the positivity \eqref{eq:cRWM>0} of the universal Nazaros-Sodin constant.

To show the connectedness of $\Bc$ we recall that $\Bc$ is closed ~\cite[Proposition 1.2]{KuWi attainable} w.r.t. taking convolutions
$\rho_{1},\rho_{2}\mapsto \rho_{1}\star \rho_{2}$,
and that there exists ~\cite[Proposition 1.2]{KKW 2013} a path $a\mapsto \nu_{a}$, $a\in [0,\pi/4]$ between $\nu_{0}$ the Cilleruelo measure and
$\nu_{\pi/4} = \frac{d\theta}{2\pi}$ the uniform measure on $\Sc^{1}$; $\nu_{a}$ is the arc-length measure on $\theta\in [-a,a]$, symmetrised
to be $\pi/2$-rotation invariant. The above implies that given $\rho\in\Bc$, we may construct a path
$$\{\rho \star \nu_{a}\}_{a\in [0,\pi/4]}$$ between $\rho$ and $\rho \star \frac{d\theta}{2\pi} = \frac{d\theta}{2\pi}$, so that
$\Bc$ is path-connected (in particular, connected).

\end{proof}

\subsection{Proof of Lemma \ref{lem:cNS(Cil)=0}: the Nazarov-Sodin constant of the Cilleruelo measure vanishes}

\label{sec:NS(Cil)=0 proof}

We give two different proofs, each independently informative; the same
proofs are valid for the tilted Cilleruelo.  The first proof uses the
fact that the limit random field can be realized explicitly as a
trigonometric polynomial with only four nonzero coefficients. The
second proof is based on a local computation of the number of
``flips'' in the direction of the line $x_{1}=x_{2}$
% ; more generally,
% the expected number of nodal points with a given normal direction was
% evaluated ~\cite{RW2017}.

\begin{proof}[Proof $1$: Limit random field]

Let $$\nu_{0} = \frac{1}{4}(\delta_{\pm 1} + \delta_{\pm i})$$ be the
Cilleruelo measure; the corresponding covariance function is  then
\begin{equation}
\label{eq:r0 covar Cil def}
r_{0}(x):= \frac{1}{2}(\cos(x_{1})+\cos(x_{2}))
\end{equation}
with $x=(x_{1},x_{2})\in\R^{2}$, and
\begin{equation*}
\E[f_{\nu_{0}}(x)\cdot f_{\nu_{0}}(y)] = r_{0}(x-y),
\end{equation*}
for $x,y\in\R^{2}$.
Let us describe the corresponding Gaussian random field
$f_{0}=f_{\nu_{0}}$ explicitly. We may realize it as
\begin{equation}
\label{eq:frho Cil 4 sum}
f_{0}(x) = \frac{1}{\sqrt{2}}\left( \xi_{1}\cos(x_{1})+\xi_{2}\sin(x_{1})+\xi_{3}\cos(x_{2})+\xi_{4}\sin(x_{2})  \right),
\end{equation}
where $(\xi_{1},\xi_{2},\xi_{3},\xi_{4})$ is a standard $4$-variate Gaussian; equivalently $\{\xi_{i}\}_{i=1}^{4}$ are standard Gaussian i.i.d.

Alternatively, we may rewrite \eqref{eq:frho Cil 4 sum} as
\begin{equation}
\label{eq:f0 Cil def}
f_{0}(x) = \frac{1}{\sqrt{2}} \cdot \left(a_{1}\cdot \cos(x_{1}+\eta_{1}) + a_{2}\cdot \cos(x_{2}+\eta_{2})   \right),
\end{equation}
where $a_{1}=\sqrt{\xi_{1}^{2}+\xi_{2}^{2}}$, $a_{2}=\sqrt{\xi_{3}^{2}+\xi_{4}^{2}}$ are $\text{Rayleigh}(1)$
distributed independent random variables (equivalently, $\chi$ with $2$ degrees of freedom), and $\eta_{1},\eta_{2}\in [0,2\pi)$ are random phases uniformly drawn in $[0,2\pi)$. Let us now determine the zero set of $f_{0}$ in \eqref{eq:f0 Cil def} on $\Tb^{2}$; we claim that $f_{0}$ has no compact nodal components at all; accordingly for every $R>0$ we have $$\Nc(f_{0};R) \equiv 0.$$

\begin{figure}[ht]
\centering
\includegraphics[height=60mm]{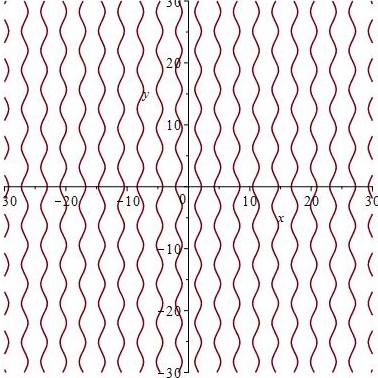}
\includegraphics[height=60mm]{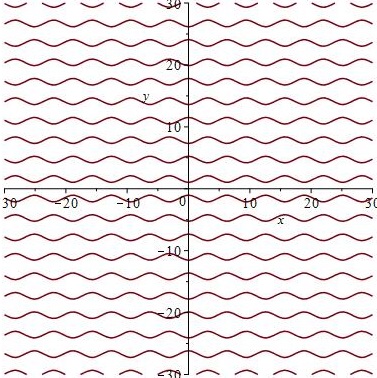}
\caption{Solution plot for equations $2\cos(x_{1})+\cos(x_{2})=0$ (left) and $\cos(x_{1})+2\cos(x_{2})=0$ (right). In both cases there are no compact nodal components.}
\label{fig:solCilhorvert}
\end{figure}

First, by translation $y = x+(\eta_{1},\eta_{2})$, we may assume that $\eta_{1}=\eta_{2}=0$, so that $f_{0}(x)=0$ if and only if
\begin{equation}
\label{eq:a1cos(y1)=-a2cos(y2)}
a_{1}\cdot \cos(y_{1}) = - a_{2}\cdot \cos(y_{2}).
\end{equation}
Now suppose that the coefficients in \eqref{eq:a1cos(y1)=-a2cos(y2)} satisfy $a_{1}>a_{2}$ (occurring with probability $\frac{1}{2}$).
Given $y_{1}$ there is a solution for $y_{2}$ to \eqref{eq:a1cos(y1)=-a2cos(y2)}, if and only if
\begin{equation}
\label{eq:y solvable Cil}
y_{1} \in  \left[\arccos\left(\frac{a_{2}}{a_{1}}\right),\pi-\arccos\left(\frac{a_{2}}{a_{1}}\right)\right] + k\pi
\end{equation}
for some $k\in\Z$. A number $y_{1}$ lying in the open interval on the r.h.s. of \eqref{eq:y solvable Cil}
corresponds to precisely two solutions for $y_{2}$ in each period $y_{2}\in [j\cdot 2\pi, (j+1)\cdot 2\pi)$
(depending on the parity of $k$ in \eqref{eq:y solvable Cil}). For the endpoints $y_{1}$ of the interval on the r.h.s. of \eqref{eq:y solvable Cil}
there exists a unique solution $y_{2} = (2j+1)\pi$ and $2j\pi$ to the left and right endpoints respectively in case $k$ in
\eqref{eq:y solvable Cil} is even, and the other way around in case $k$ is odd.
The above means that the solution curve of \eqref{eq:a1cos(y1)=-a2cos(y2)} consists of ascending oscillating periodic curves
(see Figure \ref{fig:solCilhorvert}, left) with no compact components at all. The situation when the coefficients in
\eqref{eq:a1cos(y1)=-a2cos(y2)} satisfy $a_{1}<a_{2}$ is a mirror image of the just considered (see Figure \ref{fig:solCilhorvert}, right); the event $a_{1}=a_{2}$ does almost surely not occur.

\end{proof}

\begin{proof}[Proof $2$: Local estimates]

We reuse the notation \eqref{eq:nu0 Cil meas} for the Cilleruelo measure $\nu_{0}$, the
covariance function $r_{0}(x)$ given by \eqref{eq:r0 covar Cil def}, and
$f=f_{\nu_{0}}$; also recall the notational conventions that $f_{1}
= \partial_{1} f = \partial f / \partial x_{1}$, $f_{12}
= \partial_{1} \partial_{2} f$ etc.
Given a smooth {\em closed} planar curve $\gamma:[0,L]\rightarrow\R^{2}$ and
a unit vector $\xi\in \Sc^{1}$ there
exists
a point $t\in [0,L]$ such that the tangent $\dot{\gamma}(t) = \pm \xi$
of $\gamma$ is in the direction $\pm \xi$. Therefore
\begin{equation}
\label{eq:E[N(f0,R)]<=E[flips]}
\E[\Nc(f;R)] \le \E[x\in \Dc_{R}:\: f_{0}(x)=f_{1}(x)+f_{2}(x) = 0].
\end{equation}
In what follows we will find that the
r.h.s. of \eqref{eq:E[N(f0,R)]<=E[flips]} vanish, and thus so does the
l.h.s.; this certainly implies that
$c_{NS}(\nu_{0})=0$.

To this end we
define $$F(x):=(f(x),f_{1}(x)+f_{2}(x))$$
and use Kac-Rice \cite[Theorem 6.3]{AW} to write
\begin{equation}
\label{eq:E[flips]=int(K1)}
\E[x\in B(R):\: f(x)=f_{1}(x)+f_{2}(x) = 0] = \int\limits_{B(R)}K_{1}(x)dx,
\end{equation}
where
$$
K_{1}(x)\equiv K_{1}(0) = \phi_{F(0)}\cdot \E[|J_{F}(0)| \big| F(0) ]
$$
by stationarity (see \S~\ref{sec:Kac-Rice}); $F(0)$ is
non-degenerate by the independence
of the components
of the Gaussian
vector
$$
(f(0),\nabla f(0))\in \R^{3}.
$$
Now
\begin{equation}
\label{eq:JF Jac Cil flips}
J_{F}(0)=\left(\begin{matrix}
f_{1} & f_{2} \\ f_{11}+f_{12}
& f_{12}+f_{22}
\end{matrix}\right),
\end{equation}
where all of the matrix entries are evaluated at the origin. Moreover,
a direct computation with $r_{0}$ reveals
that
$$\var(f_{12})=0,$$
i.e. $f_{12}=0$ a.s. Hence \eqref{eq:JF Jac Cil
  flips} is
\begin{equation*}
J_{F}(0)=\left(\begin{matrix}
f_{1} &f_{2} \\ f_{11}
& f_{22}
\end{matrix}\right),
\end{equation*}
and
\begin{equation*}
\det {J_{F}(0)} = f_{1}\cdot f_{22} - f_{2}\cdot f_{11};
\end{equation*}
conditioned on
$$
f(0)=f_{1}+f_{2} = 0
$$ this equals to
\begin{equation*}
\det {J_{F}(0)} =
f_{1}\cdot f_{22} + f_{1}\cdot f_{11}
= f_{1}\cdot (f_{11}+f_{22}) = -f_{1}\cdot f = 0,
\end{equation*}
since $f$ satisfies the Schr\"{o}dinger equation $\Delta f+
f=0$, and we condition on $f(0)=0$. Hence $K_{1}(x)\equiv 0$ vanishes
identically, the expectation on the l.h.s. of \eqref{eq:E[flips]=int(K1)} vanishes, which, as it was mentioned above,
is sufficient to yield the statement of Theorem \ref{thm:NS const torus}.

\end{proof}

\section{Spectral measures $\rho$ with $d_{NS}(\rho) > 0$}
\label{sec:example-spectr-meas}

We give two examples of trigonometric polynomials $f,g$, both
realisable from the same spectral measure (with support on the three
pairs of antipodal points
$\{ \pm (1,0), \pm (3,0) , \pm (0,1) \}$), namely
$$
f(x,y) =  \sin(x)+0.8\cdot \sin(3 x) + \sin(y), \quad
g(x,y) = \sin(x)+0.8\cdot \sin(3 x) + 0.2 \cdot \sin(y)
$$
where $f$ has many compact nodal domains, whereas $g$ does
not.

\begin{figure}[ht]
\centering
\includegraphics[width=75mm]{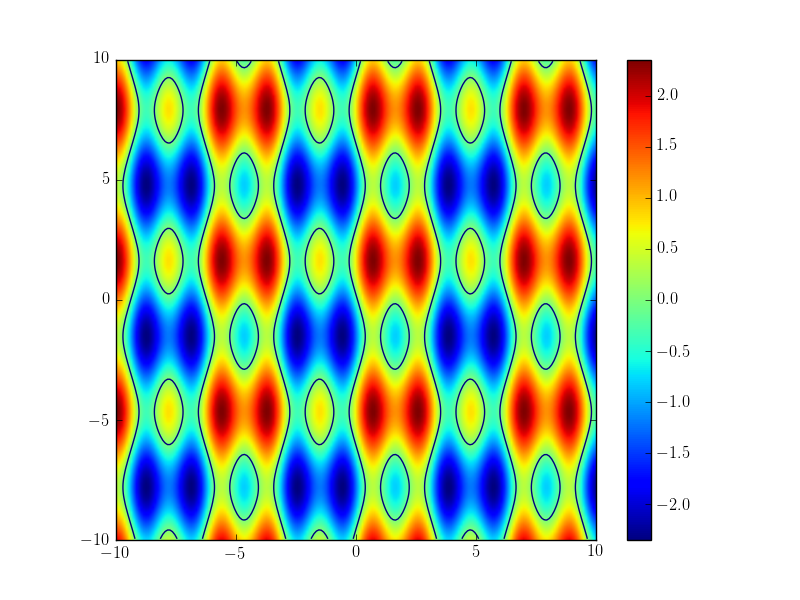}
\includegraphics[width=75mm]{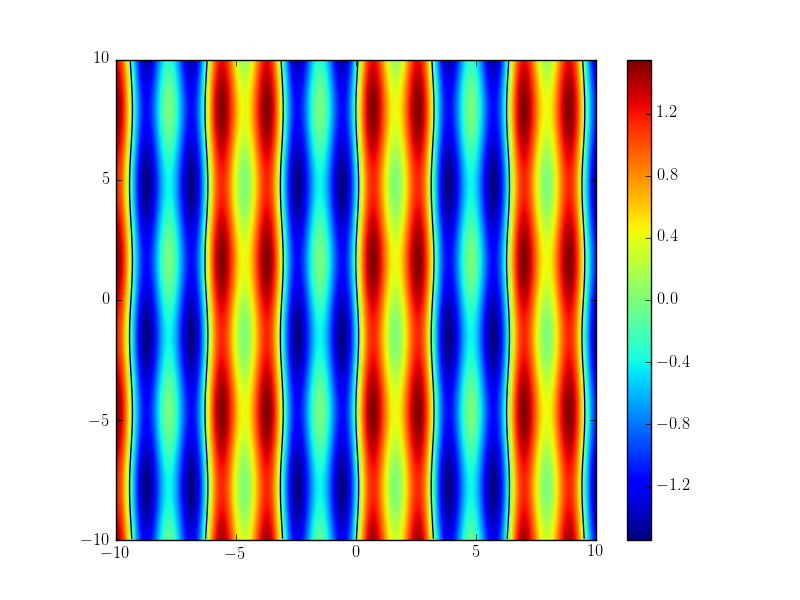}
\caption{Plot of $f(x,y)$ (left) and $g(x,y)$ (right). Note that
   compact nodal components only occurs in the left plot.  Nodal
   curves in black.}
\label{fig:positive_d}
\end{figure}
It is easy to verify that both $f$ and
$g$ are stable in the sense that the density of the number of nodal
domains (per area unit) remains the same under small perturbations of
the form
$$
f \to f + \epsilon_{1}\sin(x)+
\epsilon_{2}\cos(x)+
\epsilon_{3} \sin(3 x) +\epsilon_{4} \cos(3 x) +
\epsilon_{5} \sin(y) + \epsilon_{6} \cos(y)
$$
Thus both types of events (i.e., having no compact nodal domains or
having a positive density of compact nodal domains per area unit)
occur with probability $>0$. Hence there exists $\rho$ such that
$d_{NS}(\rho) > 0$.

The spectral measure in the above example
%is supported on only three
%pairs of antipodal points, but it
is not monochromatic, but using
%beautiful construction by Stern together with
a recent result by
M. Ingremeau we can give examples of monochromatic spectral measures
$\rho$, also supported on three pairs of antipodal points, with the property
that $d_{NS}(\rho) > 0$.
Namely, let $\rho$ be the uniform probability measure supported on the
six points $\pm (1,0), \pm (0,1), \pm (1,1)/\sqrt{2}$.
Letting
$$
g(x,y) := 2 \cos(x) + \cos(y)
$$
we find that $g$ has no compact nodal domains
(cf. Figure~\ref{fig:solCilhorvert}); it is straightforward to verify
that the gradient is non-vanishing on the nodal set of $g$.  Since $g$
is doubly periodic there exists $\beta>0$ such that
$g$ has no $\beta$-unstable points in $\R^2$
(cf. \S~\ref{sec:Omega0 constr}).  As $g$ is stable, the nodal
pattern persists for small perturbations of the form
$$
g = f + \epsilon_{1} \sin(x) +
\epsilon_{2} \sin(y) +
\epsilon_{3} \cos((x+y)/\sqrt{2}) +
\epsilon_{4} \sin((x+y)/\sqrt{2})
$$
(for $\epsilon_{1}, \ldots, \epsilon_{4}$ sufficiently small).  Hence,
given any $\epsilon >0$, there exists $R_{0}$ such that the event
$$\Nc(f_{\rho};R)/(4R^{2}) < \epsilon,$$ for all $R \geq R_{0}$, occurs
with positive probability.
On the other hand, Ingremeau (cf. \cite[Remark~3]{In}) has shown that
$c_{NS}(\rho)>0$ for any spectral measure $\rho$ with proper support
on three or more pairs of antipodal points, and hence
$d_{NS}(\rho) > 0$.

\appendix
\section{Plots of Cilleruello type eigenfunction}
\label{sec:cill-plots}
We begin with a higher resolution plot of the
eigenfunction shown in Figure~\ref{fig:cilleruello}.  As can be seen
the nodal domains tend to be either vertical or horizontal, and extend
many wavelengths.
\begin{figure}[ht]
\centering
\includegraphics[width=12cm]{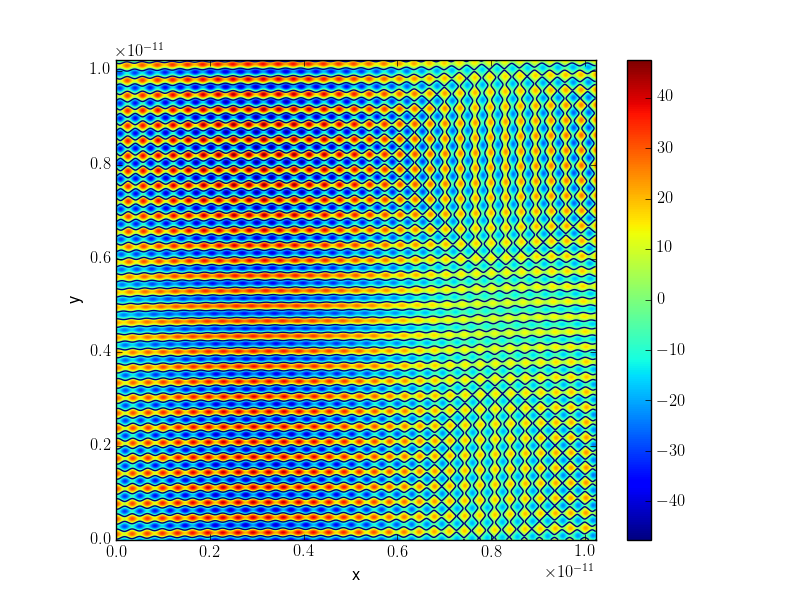}
\caption{ Fragment of a Cilleruello type eigenfunction; nodal curves in
  black as before.  Here $n=9676418088513347624474653$ and
  $r_{2}(n) = 256$.  }
\label{fig:cilleruello-dummy2}
\end{figure}

Below we give examples of the most extreme type of Cilleruello
eigenfunctions in terms of the spectral measure having smallest
possible angular support.  These arise from primes of the form
$n = a^{2}+1$; we then have $r_{2}(n) = 8$ and the set of lattice
points $\{(x,y) \in \Z^2 : x^{2}+y^{2} = n \}$ are of the form
$
\{ (a, \pm 1), (-a, \pm 1), (1, \pm a), (-1, \pm a) \}
$
and the angles between these vectors and either the $x$, or $y$,
coordinate axis is very small for $a$ large.

%\vspace{-5mm}
\begin{figure}[ht]
\centering
\includegraphics[width=12cm]{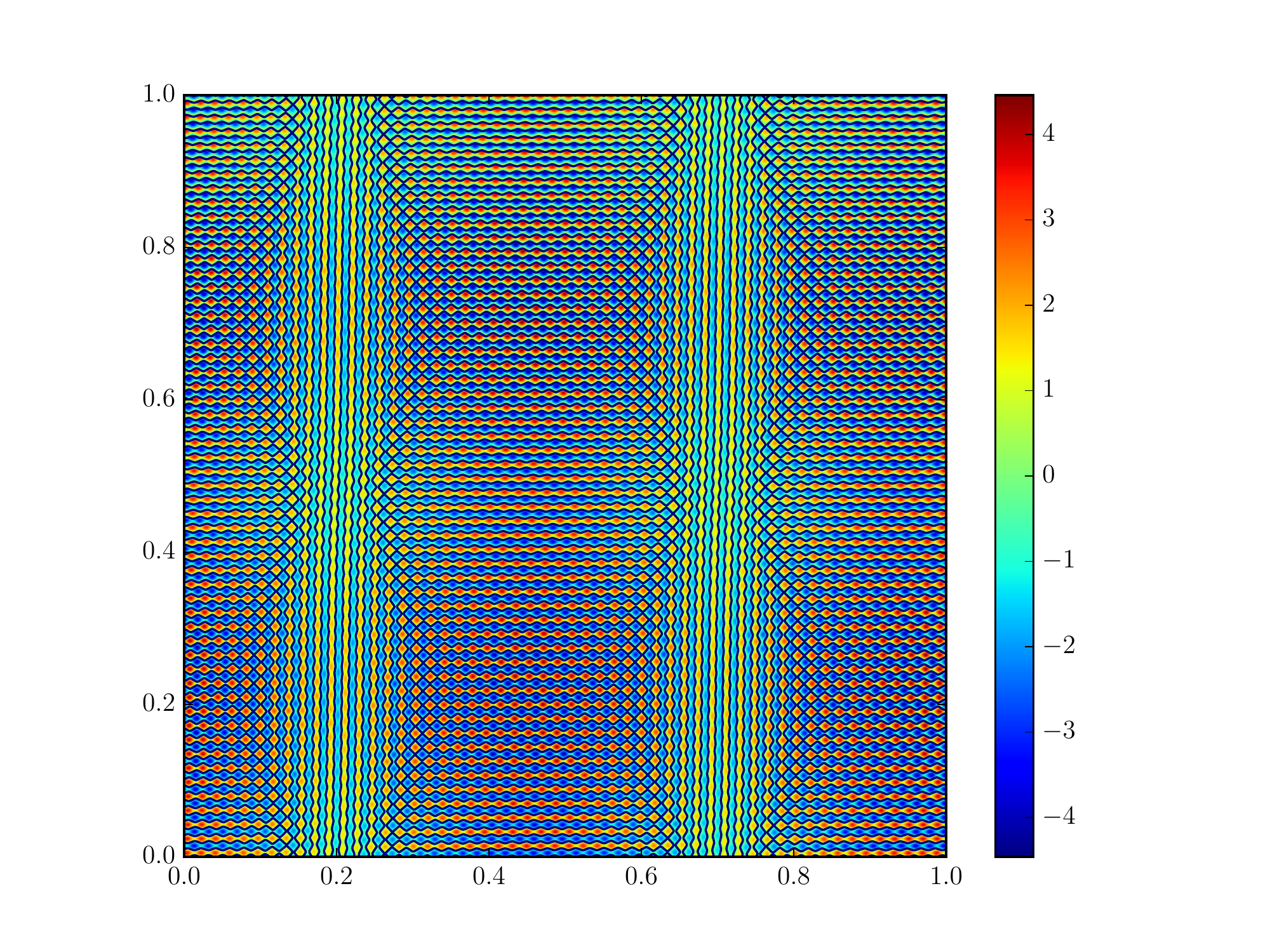}
\vspace{-3mm}
\caption{Plot of random Cilleruello type eigenfunction,
  for $n = 54^{2}+1$ and $r_{2}(n) = 8$.
}
\label{fig:cilleruello-eight0}
\end{figure}

\vspace{-5mm}
\begin{figure}[ht]
\centering
\includegraphics[width=12cm]{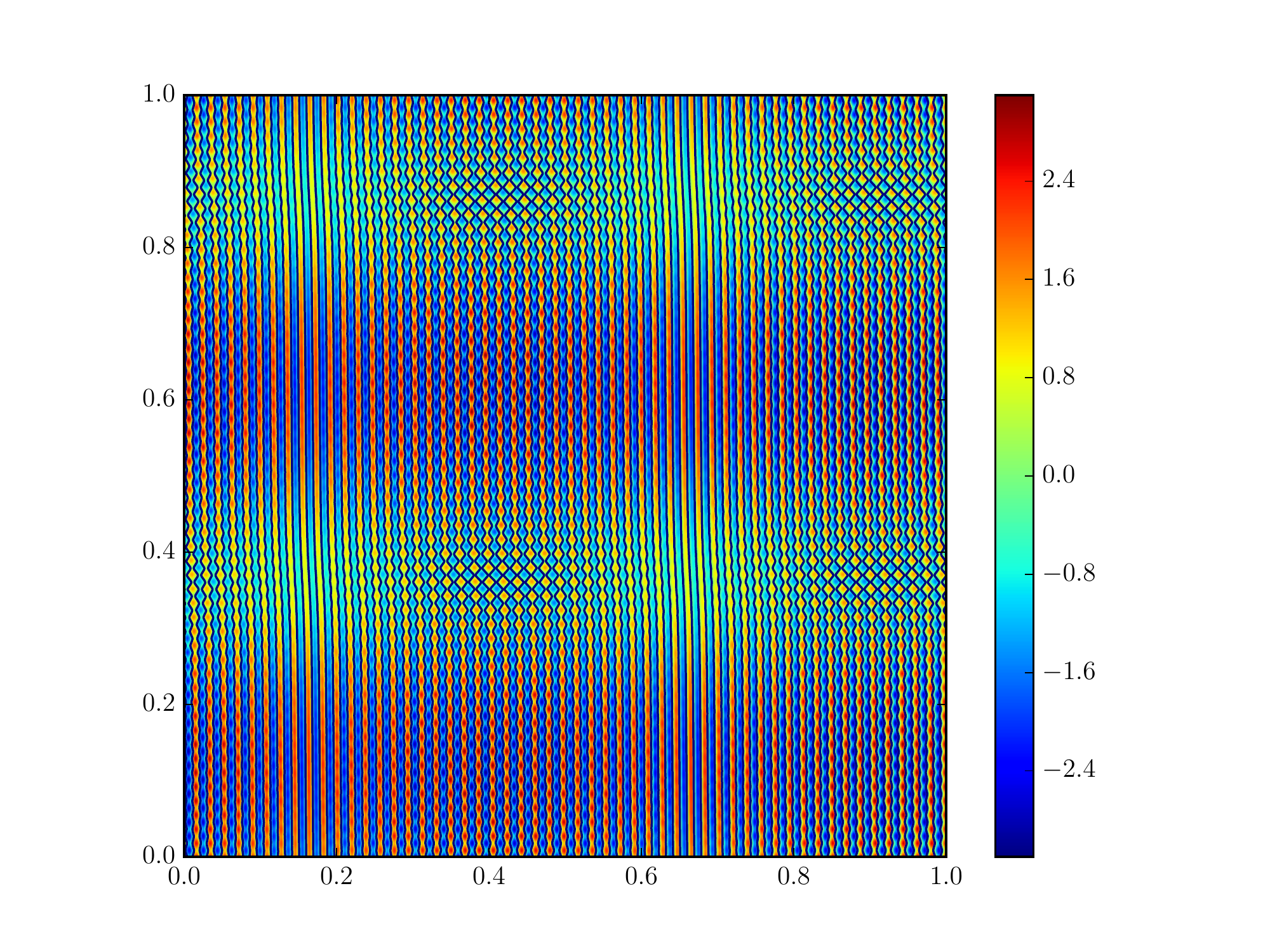}
\vspace{-3mm}
\caption{Plot of random Cilleruello type eigenfunction,
  for $n = 54^{2}+1$ and $r_{2}(n) = 8$.
}
\label{fig:cilleruello-eight1}
\end{figure}

\end{document}